\documentclass[journal,draft,onecolumn,12pt,twoside]{IEEEtran}
 
\usepackage{placeins}

\usepackage{array}
\usepackage{longtable}
\newcolumntype{L}[1]{>{\raggedright\let\newline\\\arraybackslash\hspace{0pt}}m{#1}}

\usepackage[inline]{enumitem}

\usepackage[dvipsnames]{xcolor}

\usepackage[final]{graphicx}
\usepackage{ragged2e}
\usepackage{makecell}
\usepackage{tikz}
\usetikzlibrary{backgrounds}
\usetikzlibrary{arrows}
\tikzstyle{arrow} = [thick,->,>=stealth]
\tikzset{
	block/.style = {draw, thick, minimum width=1.7cm, minimum height=1.2cm, node distance=3cm},
	down/.style={yshift=-7em},
	up/.style={yshift=7em},
	right/.style={xshift=15em},
	downright/.style={xshift=12em, yshift = -7em},
	downleft/.style={xshift=-12em, yshift = -7em}
}
\tikzset{
	circ/.style={circle, fill=black, inner sep=2pt, node contents={}}
}
\usetikzlibrary{shapes,arrows}
\tikzstyle{line}=[draw] 
\usepackage{epstopdf}
\usetikzlibrary{calc}
\usepackage[tableposition=top,font=small]{caption}
\usepackage[caption=false,font=footnotesize]{subfig}
\usetikzlibrary{arrows.meta}
\usepackage{calrsfs}
\DeclareMathAlphabet{\pazocal}{OMS}{zplm}{m}{n}

\usepackage{amssymb}
\usepackage{amsthm} 
\theoremstyle{definition}
\newtheorem{definition}{Definition}
\usepackage{amsmath}
\DeclareMathOperator*{\argmax}{arg\,max}
\DeclareMathOperator*{\argmin}{arg\,min}

\DeclareMathOperator*{\max2}{\,max}
\usepackage{mathtools}
\newcommand\givenbase[1][]{\:#1\lvert\:}
\let\given\givenbase

\DeclarePairedDelimiterX\Basics[1](){\let\given\sgiven #1}
\newtheorem{theorem}{Theorem}[section]

\newtheorem{corollary}{Corollary}[theorem]
\newtheorem{lemma}[theorem]{Lemma}
\let\norm\undefined 
\DeclarePairedDelimiter\norm{\lVert}{\rVert}
\DeclareSymbolFont{matha}{OML}{txmi}{m}{it}
\DeclareMathSymbol{\varv}{\mathord}{matha}{118}

\title{Fast reinforcement learning for decentralized MAC optimization}
\author{Eleni Nisioti and Nikolaos Thomos~\IEEEmembership{Senior~Member,~IEEE}
\thanks{E. Nisioti and N. Thomos are with the School of Computer Science and Electronic Engineering, University of Essex, Colchester, United Kingdom (e-mail:e.nisioti, nthomos@essex.ac.uk).} 
}

\begin{document}
\maketitle
\begin{abstract}
In this paper, we propose a novel decentralized framework for optimizing the transmission strategy of Irregular Repetition Slotted ALOHA (IRSA) protocol in sensor networks. We consider a hierarchical communication framework that ensures adaptivity to changing network conditions and does not require centralized control. The proposed solution is inspired by the reinforcement learning literature, and, in particular, Q-learning. To deal with sensor nodes' limited lifetime and communication range, we allow them to decide how many packet replicas to transmit considering only their own buffer state. We show that this information is sufficient and can help avoiding packets' collisions and improving the throughput significantly. We solve the problem using the decentralized partially observable Markov Decision Process (Dec-POMDP) framework, where we allow each node to decide independently of the others how many packet replicas to transmit. We enhance the proposed Q-learning based method with the concept of \textit{virtual experience}, and we theoretically and experimentally prove that convergence time is, thus, significantly reduced. The experiments prove that our method leads to large throughput gains, in particular when network traffic is heavy, and scales well with the size of the network. To comprehend the effect of the problem's nature on the learning dynamics and vice versa, we investigate the \textit{waterfall effect}, a severe degradation in performance above a particular traffic load, typical for codes-on-graphs and prove that our algorithm learns to alleviate it.
\end{abstract}

\section{Introduction}\label{sec:intro}
%
%
%
%
%

The scenery of Internet of Things (IoT) technology is rapidly evolving, both in terms of opportunities and needs, and is expanding its outreach to a wide spectrum of daily life applications. Communication in IoT networks and wireless sensor networks (WSNs) is in general challenging, as IoT devices and sensors have limited capabilities, such as limited battery capacity and communication range. To coordinate the access of the shared wireless resources, a MAC protocol is employed. MAC design aims at optimizing the performance of communication by formulating the strategies IoT or sensor nodes use to access the common channel. Communication protocols, such as Slotted ALOHA \cite{Abramson1970THEAS}, offer efficient random access mechanisms, but face problems for networks of increased size and channels with varying noise conditions and network load. Thus, there is still an urgent need to redesign ALOHA so that it optimally uses the available bandwidth and users can obtain the demanded content with fewer transmissions  and without imposing coordination between the nodes. Such optimization of Slotted ALOHA will lead to prolonging the life of the sensors as fewer transmissions will be required for the communication.  

MAC protocol design is often studied as a distributed resource allocation problem, where sensors attempt transmission of packets to a shared channel, and therefore, compete for the restricted bandwidth resources. There exist two diametrical families of MAC protocols, namely: \begin{enumerate*}[label={(\alph*)}]
\item Time-Division Multiple Access (TDMA) based, where allocation of slots is static and performed a-priori, and \item contention-based, where nodes
randomly select time slots to transmit. \end{enumerate*} TDMA has been successfully applied in VANETS \cite{Hadded2015} due to its ability to provide deterministic access time without collisions in real-time applications. Conversely, contention-based methods are more appropriate for adaptive scenarios where resources and communication load change over time and energy consumption is limited \cite{1673243}, despite the fact that in these methods packet collisions occur because of the random packet transmission decisions made.  

Slotted ALOHA, belonging to the family of contention-based protocols, is widely used for designing random multiple access mechanisms, but suffers from low throughput due to packet collisions that lead to packet loss. Diversity Slotted ALOHA (DSA) \cite{Choudhury1983} significantly improves upon it by introducing a burst repetition rate, that allows network nodes to transmit a pre-defined number of replicas of the original messages. The introduction of the repetition rate enables Contention Resolution Slotted ALOHA (CRDSA) \cite{4155680}, that helps exploiting interference cancellation (IC) for the retrieval of collided packets. To further improve the performance of \cite{Choudhury1983,4155680}, Irregular Repetition Slotted ALOHA (IRSA), introduced in \cite{5668922}, allows for a variable number of replicas for each user. The work in \cite{5668922}  relates the process of successive interference cancellation applied to colliding users to the process of iterative belief-propagation (BP) erasure-decoding of codes-on-graphs.  The number of replicas in IRSA is decided by sampling from a probability distribution, which is designed such as to decrease packet loss. IRSA shows that diversity in the behavior of individual nodes, in the form of selecting the number of replicas, results in better overall throughput. 

Further improvements of IRSA can be found in the work of \cite{7302046}, which extends IRSA by introducing Coded Slotted ALOHA (CSA), where coding is performed between the packets available at the nodes. In \cite{6503624}, a frameless variant of CRDSA is introduced, which limits delays, as sensor nodes are not obliged to wait for the next frame to transmit their messages. Frame asynchronous Coded ALOHA \cite{7762138} combines methods in \cite{7302046} and \cite{6503624} and shows an improvement  both in achieved error floor and observed delay. Although these are interesting research directions, computational complexity introduced because of the coding procedure compared to the non-coding variants may limit their use in the sensor networks under study. For this reason, we do not explore this direction, but we leave it as a future work. However, we should note that our scheme is generic and can be extended for the coded variants of IRSA. 

IRSA performance depends on the optimization scheme used to derive the degree distribution function, i.e., the probability distribution used to decide the number of replicas. This distribution can be optimized using differential evolution, which is used to asymptotically analyze the transmission policy, i.e., the number of replicas. More recently in \cite{Toni2018}, the use of Multi-armed Bandits (MABs) was introduced, as a remedy for inaccurate asymptotic analysis in non-asymptotic settings and as an alternative to computationally expensive finite length block analysis. This work has been proposed for an IRSA variant that incorporates users' prioritization \cite{7305777}. The main drawback of this formulation is that it leads to a continuous action space, an intractability addressed through discretization, that has been proven to significantly degrade performance \cite{Waugh:2009:APE:1558109.1558119}. Another disadvantage of MABs is that their framework is not expressive enough as they are stateless. This renders MABs inappropriate for sensor networks, where operations are constrained by sensor
nodes' characteristics, such as battery level, memory size, etc., valuable information that MABs fail to incorporate in the decision-making process.

In this paper, we investigate the optimization of the transmission strategy of sensor networks following the Markov Decision Process (MDP) \cite{Bellman:2003:DP:862270} framework. In particular, in our scheme sensor nodes are capable of independently and distributively learning the optimal number of replicas to transmit in a slotted IRSA protocol. Guided by the nature of the problem under consideration, we design a distributed, model-free, off-line learning algorithm that deals with partial observability, which refers to the inability of a sensor node to observe information that requires global access to the network. Hence, under partial observability nodes act on information only local to them, for example the state of their buffer, i.e., the number of packets in it. This approach has successfully been applied in the domain of sensor networks \cite{AAAI113750} and draws from its need for scalable, efficient, decentralized optimization algorithms. To deal with partial observability we employ decentralized POMDP (Dec-POMDP) algorithms, that are associated with high complexity, as they are NEXP-Complete \cite{DBLP:journals/corr/abs-1301-3836}. Hence, to overcome this problem, we explore realistic variations of it that exploit the problem's nature, in particular independence of agents in terms of learning. Distributed optimization in sensor networks has been extensively studied in \cite{Lesser:2003:DSN:940763} and successful applications have mainly been offered in the areas of packet routing \cite{1420665} and object tracking \cite{Nair:2005:NDP:1619332.1619356}. Machine learning concepts have been explored in \cite{1673243}, where an actor-critic algorithm to optimally schedule active times in an Timeout-MAC protocol is presented and \cite{DBLP:journals/corr/abs-1710-08803}, where a multi-state sequential learning algorithm is proposed, that learns the number of existing critical messages and reallocates resources in a contention-free MAC protocol. However, none of these works addresses decentralized resource allocation under a random access MAC mechanism. Our solution leverages techniques from the multi-agent reinforcement learning literature to design transmission strategies for agents that optimally manage the available time slots and maximize packet throughput. To the best of our knowledge, this is the first attempt to formulate a decentralized and adaptive solution for MAC design in the context of Slotted ALOHA. Our main contributions consists in:
\begin{itemize}[noitemsep,,topsep=0pt]
\item the design of an intelligent sensor network that adapts to communication conditions and optimizes its behavior in terms of packet transmission, using reinforcement learning;
\item the derivation of an algorithm from the family of Dec-POMDP that employs virtual experience concepts to accelerate the learning process \cite{MnihKSGAWR13};
\item the investigation of the impact of the waterfall effect on the learning dynamics and the ability of our proposed algorithm to alleviate it. 
\end{itemize}
Section \ref{sec:framework} describes the problem under investigation, introduces the suggested framework and models the problem highlighting underlying assumptions. In Section \ref{sec:IRSA}, we provide the necessary theoretical background by outlining the vanilla IRSA protocol in order to derive the goal of optimization. In Section \ref{sec:protocol}, we formulate our proposed decentralized reinforcement learning based POMDP IRSA protocol, henceforth referred to as Dec-RL IRSA. Finally, Section \ref{sec:experiments} exhibits the experiments performed to configure and evaluate our optimization technique. 

\section{Intelligent sensor network framework} \label{sec:framework}

%
%
\subsection{Sensor network description}

Let us consider a network of $M$ sensor nodes collecting measurements from their environment and transmitting them to a core network for further process. The main bottleneck of the operation of the network is the transmission of the packets nodes possess through a common communication medium, as it is also used by neighboring sensor nodes that transmit their packets over it. Abiding to the vanilla Slotted ALOHA framework and its variants, in our work time is divided into frames of fixed duration, each one consisting of $N$ time slots. At the beginning of each frame each sensor randomly chooses one of the $N$ available slots to transmit its packet. ALOHA transmission protocol is depicted in Fig. \ref{fig:aloha}. In this paper, a contention-based approach is used, and, therefore, collisions occur due to the fact that sensors may choose to transmit simultaneously in a slot. This results in a degradation of the observed throughput.

\subsection{Proposed communication framework}\label{sec:sframework}

The design of an efficient MAC protocol requires sensor nodes to be equipped with the capability of independently deciding upon their transmission strategy. Traditional approaches solve the MAC optimization centrally, assuming that all problem-related information will become available to a central node. This introduces a communication overhead that is needed to exchange the information required to make the optimal transmission decisions. This communication is expensive for large-sized networks, and, in general, does not scale well with the size of the network. Further, centralized algorithms fail to exploit the underlying network structure, which can facilitate the optimization of transmission strategies by exhibiting characteristics such as locality of interaction. Here, we aim at designing a protocol that can be easily applied in large-sized networks, as well as to optimize its functionality in a distributed way, considering sensor nodes as the basic building block. In Fig. \ref{fig:network}, we illustrate the overall structure of our communication model. This resulted from the following desired characteristics:
\begin{figure}
\begin{minipage}[t]{0.35\textwidth}
\centering
\resizebox{!}{3cm}{
\input{sensors.tex}}
\caption*{(a)}
\end{minipage} %
\begin{minipage}[t]{0.6\textwidth}
\centering
\tikzstyle{background grid}=[draw, black!50,step=1cm]

\begin{tikzpicture}
\node[text width = 1.2cm, text height=0.5cm, minimum width = 1.2cm, align = center] (u1) at (current page.north west){\small User 1};
\draw[help lines,step=0.5cm,xstep=1cm,dashed] ([xshift=1cm, yshift=0.15cm]u1) grid ([xshift=7cm,yshift=-2.5cm]current page.north west);
\node[text width = 1.2cm, text height=0.5cm, minimum width = 1.2cm, align = center] (u2) at ([yshift=-0.5cm]u1){\small User 2};
\node[text width = 1.2cm, text height=0.5cm, minimum width = 1.2cm, align = center] (u3) at ([yshift=-0.5cm]u2){\small User 3};
\node[text width = 1.2cm, text height=0.5cm, minimum width = 1.2cm, align = center] (dots) at ([yshift=-0.5cm]u3){\small \vdots};
\node[text width = 1.2cm, text height=0.5cm, minimum width = 1.2cm, align = center] (um) at ([yshift=-0.5cm]dots){\small  User m};
\node[draw, minimum width = 0.4cm, dashed, red] (p_1) at([xshift= 1.5cm,yshift = -0.2cm]u1) {};
\node[draw, minimum width = 0.4cm] (p1) at([xshift= 2.5cm,yshift = -0.2cm]u1) {};
\node[draw, minimum width = 0.4cm] (p2) at([xshift= 4.5cm,yshift = -0.2cm]u2) {};
\node[draw, minimum width = 0.4cm, dashed, red] (p_2) at([xshift= 2.5cm,yshift = -0.2cm]u2) {};
\node[draw, minimum width = 0.4cm] (p3) at([xshift= 6.5cm,yshift = -0.2cm]u3) {};
\node[draw, minimum width = 0.4cm, dashed, red] (p_3) at([xshift= 4.5cm,yshift = -0.2cm]u3) {};
\node[draw, minimum width = 0.4cm] (p4) at([xshift= 6.5cm,yshift = -0.2cm]um) {};
\node[xshift=1.2cm] (col) at($(p3)!0.5!(p4)$) {\small collision};
\draw[arrow,shorten >=0.1cm] (col.north) to (p3.east);
\draw[arrow,shorten >=0.1cm] (col.south) to (p4.east);
\draw[latex'-latex',double] ([yshift=0.3cm, xshift = 1cm]current page.north west) -- ([xshift=7cm,yshift=0.3cm]current page.north west);
\node[text width = 6cm, align = center] (fr) at ([xshift=3.5cm, yshift=0.5cm]current page.north west) {Frame, $N$ slots};
\draw[line, dashed] (p_1) -- (p1);
\draw[line, dashed] (p_2) -- (p2);
\draw[line, dashed] (p_3) -- (p3);
\draw[line, blue, dotted, thick] (p1) -- (p_2);
\draw[line, blue, dotted, thick] (p2) -- (p_3);
\draw[line, blue, dotted, thick] (p3) -- (p4);
\end{tikzpicture}
\caption*{(b)}
\end{minipage}
\caption{(a) a sensor network with a common access medium (b) transmission under the ALOHA protocol. Packets with solid lines illustrate a collision under vanilla Slotted ALOHA. The IRSA mechanism is superimposed using packets with red, dashed outlines. A black, dashed outline represents decoding of a packet and a blue, dotted line indicates that IC is performed.}
\label{fig:aloha}
\end{figure}

\paragraph{Hierarchical structure} It has been often argued that intelligent behavior of complex systems should be pursued through the adoption of hierarchical structures that support the emergence of collective intelligence \cite{4066245}. Early in the pursuit of artificial intelligence \cite{4066245} collective intelligence was recognized as a means of achieving intelligent behavior in complex systems based on interaction in populations of agents instead of sophisticated units. Inspired by \cite{4221491}, the network is organized into clusters, based on features such as proximity, common characteristics, e.g., priority or common behavior, e.g. packet content. Each cluster in Fig. \ref{fig:network}, illustrated with a dashed ellipsis, is formed by the sensors in it, one of which is the cluster-head. The latter is responsible to collect the packets from all the sensor nodes in a cluster and then transmit them to the core network. Therefore, cluster-heads serve as intermediate nodes between the sensor nodes and the core network. This design enables scalability of the network architecture. It also presents the opportunity of forming clusters based on common characteristics that affect optimization, e.g., requests for the same content can be addressed by optimizing locally the cached content in the cluster-heads. In the rest of the paper, we do not deal with the cluster formation problem but we assume that this has already taken place. Thus, we focus on the optimization of the transmission strategies of the cluster-heads.  
\paragraph{Adaptivity} Sensor networks that employ reinforcement learning to adjust to changes of their environment have been shown to be a promising approach that can ensure real-time, optimal allocation of resources in non-stationary environments \cite{1673243}. Motivated by this, our protocol is based on Q-learning, a model-free algorithm that learns optimal policies through interaction with its environment, and, thus,  adapts to its changes.  
\paragraph{Decentralization} Here, we aim at designing a decentralized solution that exploits the ad hoc, time-varying and heterogeneous nature of the network. By removing the need for a centralized point of control, our solution leverages locality of information and interaction to create nodes that contribute to the optimal overall throughput following computationally efficient policies.
\begin{figure}[t]
\centering
\scalebox{0.4}{	\begin{tikzpicture}	
	\node[block, align = center, text width=1.8cm, text centered, minimum width = 3 cm] (CN) at (0,0) {Core Network};
    \node[block,ellipse, minimum width = 12cm, minimum height = 5.2cm ] (el1) at ( [xshift = -4cm, yshift=-4cm]CN) {};
     \node[block,ellipse, minimum width = 4.6cm, minimum height = 4.4cm, dashed ] (el3) at ( [xshift = -2.5cm, yshift=-1]el1) {};
     \node[block,ellipse, minimum width = 4.5cm, minimum height = 4.3cm, dashed ] (el4) at ( [xshift = 2.3cm, yshift=1]el1) {};
    \node[block, align = center, text width=1cm, text centered, minimum width = 1 cm, minimum height=1cm] (N1) at ([xshift=-1.5cm, yshift=-1cm]el1) {Node};
    \node[block, align = center, text width=1cm, text centered, minimum width = 1 cm, minimum height=1cm] (N2) at ([xshift=-3.5cm, yshift=-1cm]el1) {Node};
    \node[block, align = center, text width=1cm, text centered, minimum width = 1 cm, minimum height=1cm] (N3) at ([xshift=-3.5cm, yshift=1cm]el1) {Node};
    \node[block, align = center, text width=1cm, text centered, minimum width = 1 cm, minimum height=1cm] (N4) at ([xshift=-1.5cm, yshift=1cm]el1) {CH};
    \node[block, align = center, text width=1cm, text centered, minimum width = 1 cm, minimum height=1cm] (N5) at ([xshift=1.5cm, yshift=1cm]el1) {Node};
     \node[block, align = center, text width=1cm, text centered, minimum width = 1 cm, minimum height=1cm] (N6) at ([xshift=3.3cm, yshift=.7cm]el1) {Node};
    \node[block, align = center, text width=1cm, text centered, minimum width = 1 cm, minimum height=1cm] (N7) at ([xshift=1.5cm, yshift=-1cm]el1) {CH};
    \node[block,ellipse, minimum width = 10cm, minimum height = 5.2cm ] (el2) at ( [xshift = 4cm, yshift=4cm]CN) {};
    \node[block,ellipse, minimum width = 2.1cm, minimum height = 4.3cm, dashed ] (el5) at ( [xshift = -1.5cm, yshift=-1]el2) {};
     \node[block,ellipse, minimum width = 4.5cm, minimum height = 4.1cm, dashed ] (el6) at ( [xshift = 2.2cm, yshift=0.9]el2) {};
      \node[block, align = center, text width=1cm, text centered, minimum width = 1 cm, minimum height=1cm] (N8) at ([xshift=-1.5cm, yshift=-1cm]el2) {Node};
    \node[block, align = center, text width=1cm, text centered, minimum width = 1 cm, minimum height=1cm] (N9) at ([xshift=-1.5cm, yshift=1cm]el2) {CH};
    \node[block, align = center, text width=1cm, text centered, minimum width = 1 cm, minimum height=1cm] (N10) at ([xshift=1.5cm, yshift=1cm]el2) {CH};
    \node[block, align = center, text width=1cm, text centered, minimum width = 1 cm, minimum height=1cm] (N11) at ([xshift=1.5cm, yshift=-1cm]el2) {Node};
     \node[block, align = center, text width=1cm, text centered, minimum width = 1 cm, minimum height=1cm] (N12) at ([xshift=3.4cm, yshift=-0.5cm]el2) {Node};
     \draw[shorten >=0.1cm,->, thick,-{Latex[length=2mm]}] (N4) -- node [text width =2cm,near end,yshift=0.5cm] {Packets}  (CN.west);
     \draw[shorten >=0.1cm,->, thick,-{Latex[length=2mm]}] (N7) -- node [text width =2cm,near end,xshift=1.5cm,yshift=0.2cm] {Packets}  (CN); 
     \draw[shorten >=0.1cm,->, thick,-{Latex[length=2mm]}] (N9) -- node [text width =2cm,near end,xshift=-0.5cm] {Packets}  (CN);
     \draw[shorten >=0.1cm,->, thick,-{Latex[length=2mm]}] (N10) -- node [text width =2cm,near end,xshift=1cm,yshift=-0.5cm] {Packets}  (CN.east); 
     \draw[shorten >=0.1cm,shorten <=0.1cm,->,dashed,-{Latex[length=2mm]}] (N1) -- node [text width =2cm,midway, right, xshift=-4cm] {}  (N4);
     \draw[shorten >=0.1cm,shorten <=0.1cm,->, dashed,-{Latex[length=2mm]}] (N4) -- node [text width =2cm,midway, right, xshift=-4cm] {}  (N1);
      \draw[shorten >=0.1cm,shorten <=0.1cm,->,dashed,-{Latex[length=2mm]}] (N2) -- node [text width =2cm,midway, right, xshift=-4cm] {}  (N4);
       \draw[shorten >=0.1cm,shorten <=0.1cm,->,dashed,-{Latex[length=2mm]}] (N4) -- node [text width =2cm,midway, right, xshift=-4cm] {}  (N2);
       \draw[shorten >=0.1cm,shorten <=0.1cm,->,dashed,-{Latex[length=2mm]}] (N3) -- node [text width =2cm,midway, right, xshift=-4cm] {}  (N4);
 \draw[shorten >=0.1cm,shorten <=0.1cm,->,dashed,-{Latex[length=2mm]}] (N4) -- node [text width =2cm,midway, right, xshift=-4cm] {}  (N3);
      \draw[shorten >=0.1cm,shorten <=0.1cm,->,dashed,-{Latex[length=2mm]}] (N5) -- node [text width =2cm,midway, right, xshift=-4cm] {}  (N7);
     \draw[shorten >=0.1cm,shorten <=0.1cm,->,dashed,-{Latex[length=2mm]}] (N7) -- node [text width =2cm,midway, right, xshift=-4cm] {}  (N5);
     \draw[shorten >=0.1cm,shorten <=0.1cm,->,dashed,-{Latex[length=2mm]}] (N7) -- node [text width =2cm,midway, right, xshift=-4cm] {}  (N6);
     \draw[shorten >=0.1cm,shorten <=0.1cm,->, dashed,-{Latex[length=2mm]}] (N6) -- node [text width =2cm,midway, right, xshift=-4cm] {}  (N7);
      \draw[shorten >=0.1cm,shorten <=0.1cm,->,dashed,-{Latex[length=2mm]}] (N8) -- node [text width =2cm,midway, right, xshift=-4cm] {}  (N9);
      \draw[shorten >=0.1cm,shorten <=0.1cm,->,dashed,-{Latex[length=2mm]}] (N9) -- node [text width =2cm,midway, right, xshift=-4cm] {}  (N8);
      \draw[shorten >=0.1cm,shorten <=0.1cm,->,dashed,-{Latex[length=2mm]}] (N11) -- node [text width =2cm,midway, right, xshift=-4cm] {}  (N10);
    \draw[shorten >=0.1cm,shorten <=0.1cm,->,dashed,-{Latex[length=2mm]}] (N10) -- node [text width =2cm,midway, right, xshift=-4cm] {}  (N11);
      \draw[shorten >=0.1cm,shorten <=0.1cm,->,dashed,-{Latex[length=2mm]}] (N12) -- node [text width =2cm,midway, right, xshift=-4cm] {}  (N10);
      \draw[shorten >=0.1cm,shorten <=0.1cm,->,dashed,-{Latex[length=2mm]}] (N10) -- node [text width =2cm,midway, right, xshift=-4cm] {}  (N12);
       
    \node[text width=4cm] (t1) at ([xshift = 8cm, yshift = -2cm]el1) 
    {Sensor Network 1};
    \node[text width=4cm] (t2) at ([xshift = 7cm, yshift = -2cm]el2) 
    {Sensor Network 2};
     \node[text width=3cm] (t3) at ([xshift = -2cm, yshift = 3cm]el3) 
    {Cluster 1};
     \node[text width=3cm] (t4) at ([xshift = 5cm, yshift = 2cm]el4) 
    {Cluster 2};
    \node[text width=3cm] (t5) at ([xshift = -2cm, yshift = 3cm]el5) 
    {Cluster 1};
     \node[text width=3cm] (t6) at ([xshift = 5cm, yshift = 2cm]el6) 
    {Cluster 2};
     \draw[shorten >=0.1cm,->,dashed] (t1) -- (el1);
     \draw[shorten >=0.1cm,->,dashed] (t2) -- (el2);
     \draw[shorten >=0.1cm,->,dashed] (t3) -- (el3);
     \draw[shorten >=0.1cm,->,dashed] (t4) -- (el4);
     \draw[shorten >=0.1cm,->,dashed] (t5) -- (el5);
     \draw[shorten >=0.1cm,->,dashed] (t6) -- (el6);     
	\end{tikzpicture}}
\caption{Intelligent sensor network: Sub-networks are organized into clusters with a cluster-head (CH). Bidirectional arrows indicate exchange of information for cluster formation and sharing of each node's packets with the cluster-head, while solid lines show communication of cluster heads with the core network.}
\label{fig:network}
\end{figure}
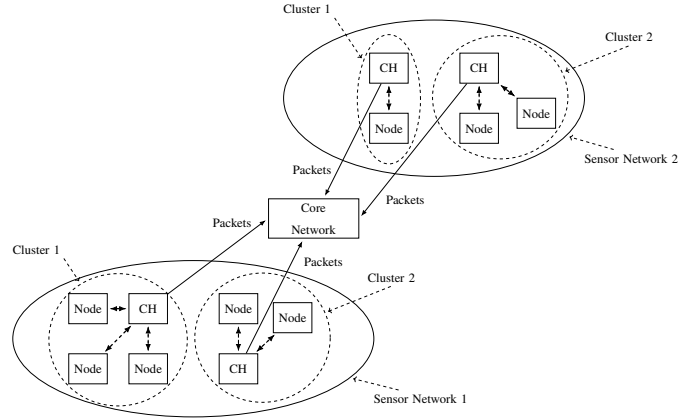

\subsection{Preliminaries}\label{sec:prelim}

This section will present our formulation and assumptions made regarding the physical layer, sensor nodes' buffers' models and packets' arrival process. Tables \ref{table:system} and \ref{table:node} summarize the notation used regarding system-related and node-related variables, respectively.

\subsubsection{Physical layer}\label{sec:physical}

We consider frequency, non-selective channels, which are characterized at the beginning of each time frame by the traffic $G_t$, where $t$ is the time index indicating the beginning of a frame. A set of slots form a frame. We assume that the traffic can be estimated perfectly in light of the number of nodes $M$ and size of frame $N$ and that it remains constant during a frame, similar to the work in \cite{7302046, 7247494}, where $G$ also stays constant for all frames. Although traffic imposes some central knowledge about the network condition, it can be easily derived by the cluster-head if we assume group-wise observability, as in the work of \cite{AAAI113750}, which denotes the ability of an agent to fully determine global information based only on observations of its cluster. 

\begin{table}
\begin{minipage}[t]{0.5\textwidth}
\centering
\caption{System-related variables}
\begin{tabular}{ |c|c| } 
 \hline
 Notation & System-related \\
 \hline
$M$ & number of sensor nodes  \\ 
$N$ & number of slots in frame \\
$G$ & channel load \\
$T$ & packet throughput \\
$K$ & total number of transmitted packets in frame\\
$PLR$ & probability loss rate\\
$F$ & size of packet \\
$S_u$& uncontrolled state\\
$E$ & number of episodes\\
$L_E$ & number of iterations in episode\\
$\alpha$ & learning rate\\
$\gamma$ & discount factor\\
$w$ & history window \\
$I$ & initial state distribution \\
$L$ & coverage time \\
$\phi$ & exponent of learning rate\\
$E$ & threshold for e-greedy exploration\\
$\pazocal{T}$ & virtual experience transformation\\
 \hline
\end{tabular}
\label{table:system}
\end{minipage} %
\begin{minipage}[t]{0.5\textwidth}
\centering
\caption{Node-related variables}
\begin{tabular}{|c|c|}
\hline
Notation & Node-related \\
\hline
$C_t$ & condition \\
$l$ & number of replicas to transmit \\
$F_t$ & number of arrivals in node's buffer \\
$B$ & size of node's buffer \\
$d$ & maximum number of replicas \\
$b_t$ & current state of buffer \\
$\Lambda(x)$ & node-degree probability distribution\\
$\pazocal{S}$ & space of states \\
$\pazocal{A}$ & space of actions\\
${\Omega}$ & space of observations\\
$\widetilde{\pazocal{H}}$ & space of virtual experience\\
$\pazocal{H}$ & space of histories\\
$R_t$ & immediate reward\\
$\rho_t$ & expected reward\\
$\pi(s)$ & policy \\
$Q^{\pi}(s,a)$ & state-action value function under policy $\pi$ \\
$V^{\pi}(s)$ & state value function under policy $\pi$ \\
\hline
\end{tabular}
\label{table:node}
\end{minipage}
\end{table}

Note that the proposed framework is oblivious to the underlying modulation and coding schemes. Similar to the work in \cite{5986747}, our only assumption is that the packet throughput and number of transmitted packets can be expressed as
\begin{align}
T_t &= T(G_{t-1}, K_{t-1}, PLR_{t-1}) \label{eq:throughput}\\
K_t &= K(G_{t}, T_{t}, PLR_{t}, C_{t})
\label{eq:action}
\end{align}
where $K_t$ represents the number of packets transmitted by all nodes during the particular frame, $PLR_t$ is the packet loss rate and $C_t$ is a node's condition. Nodes' condition can in general depend on the buffer state, battery level and, in general, any information that should potentially affect their behavior. Our work considers only the buffer state of the network nodes, as incorporating more variables in $C_t$ will increase computational requirements. However, our framework is general, and, depending on the application of interest, can easily incorporate additional characteristics to $C_t$. 

There are three sources of packet loss, i.e., packet collisions, imperfect interference cancellation and bit-level channel noise, that depends on nodes' transmission and noise power. In the rest of our work, and without loss of generality, we will assume that interference cancellation is perfect and that noise power is zero. Successful transmission will therefore be guaranteed if the iterative BP erasure-decoding algorithm, used for SIC, succeeds to recover the original packets. Thus, in our approach, (\ref{eq:throughput}) and (\ref{eq:action}) will be oblivious to $PLR_{t}$.

As suggested by (\ref{eq:throughput}), we assume that the channel throughput can be calculated in terms of the channel state, the number of transmitted packets and the packet loss probability. From (\ref{eq:throughput}), we can see that $T_t$ is a non-deterministic function of its arguments, as nodes randomly select the slots to transmit.  

\subsubsection{Buffer and traffic model}

We assume that the transmission buffer of a node is modeled as a first-in first-out queue. The source injects $F_t$ packets of size $F$ bits into the transmission buffer in each time frame according to an independent and identically distributed (i.i.d) distribution $p^f(f)$. The packets arriving to a node are stored in a finite-length buffer, of capacity $B$. Therefore, the buffer state $b^i \in \pazocal{B} = \{ 0,1, ..., B \}$ of a sensor node $i$ evolves, recursively, as follows:
\begin{align}
\begin{split}
b_0^i &= b_{init}^i\\
b_{t+1}^i &= \min\{b_t^i - T_t^i(PLR_t, K_t, G_t) + F_t^i, B^i\}
\end{split}
\end{align}
where $b_{init}^i$ denotes the initial buffer state and $T_t^i(PLR_t, K_t, G_t)$ is the packet goodput, representing the number of successfully transmitted packets in a frame for node $i$. The packets arriving after the beginning of frame $t$ cannot be transmitted until frame $t + 1$ and unsuccessfully transmitted packets stay in the transmission buffer for later retransmission.

In the following sections, we will formulate MAC optimization as a multi-agent problem and propose an efficient reinforcement learning based algorithm that enables sensor nodes to maximize the overall packet throughput of the network.

\section{IRSA overview}\label{sec:IRSA} 

%
%
In this section, we briefly overview the IRSA protocol \cite{5668922}. IRSA has been proposed to deal with the case where $M$ nodes attempt to transmit their packets into a number of transmission slots over the same communication channel. We assume that there are $N$ time slots per frame. The channel is fully characterized by its normalized traffic, defined as $G= M/N$, which represents the average number of attempted packet transmissions by all nodes per time slot. The objective of IRSA is to optimize the normalized throughput $T$, defined as the probability of successful packet transmission per slot. At the beginning of each time frame a user attempts transmission of a message by randomly choosing one of the $N$ slots to transmit a packet. In a vanilla Slotted ALOHA protocol, a transmission is successful only if no other user transmits in the same slot. The resulting throughput is a function of the normalized traffic, in particular it is $T(G) = Ge^{-G}$. In an IRSA protocol, however, a user has the capability of transmitting a variable number of replicas of the original message  in the available time slots, a strategy that improves throughput due to interference cancellation. The throughput in this case is governed by the degree distribution, a polynomial probability distribution describing the probability $\Lambda_l$ that each user transmits $l$ replicas of its message at a particular time frame. This probability distribution is expressed as 
\begin{equation}
\label{eq:distribution}
\Lambda(x) \triangleq \sum_{l=1}^{d} \Lambda_l x^l 
\end{equation}
where $d$ is the maximum number of replicas a sensor node is allowed to send. The objective of a MAC optimization algorithm is to select the values $ \Lambda_l$  in (\ref{eq:distribution}) so that overall network throughput is maximized. Formally, the optimization objective can be cast as
\begin{equation}
\label{eq:opt}
\begin{aligned}
& \text{Find:}
& &(\Lambda^{*}): \argmax_{\Lambda} T(\Lambda)\\
& \text{subject to}
& & \sum_{i=1}^{d} \Lambda_i(x) =1.
\end{aligned}
\end{equation}
The optimization in (\ref{eq:opt}) can be performed using any linear programming or gradient-based optimization algorithm, but differential evolution is usually performed \cite{5668922,7305777}. In asymptotic settings ($M \rightarrow \infty$) iterative IC convergence analysis can be used to formulate how collision resolution probability evolves with decoding time \cite{5668922} and a stability condition can be formed, which defines the maximum channel load, $G^*$, for which the probability of unsuccessful transmissions is negligible. Section \ref{sec:protocol} presents the proposed approach that allows users to learn their transmission strategies in a distributed manner for non-asymptotic scenarios. 

 \section{Dec-RL MAC protocol design} \label{sec:protocol}
 
%
%
%
%
The discussion will proceed with the adoption of the MDP model for the design of an efficient MAC optimization strategy abiding to the framework defined in Section \ref{sec:framework}. Our method employs ideas and tools from reinforcement learning and DCOP to satisfy the desired traits of the considered network setting.

\subsection{MDP formulation} Recall from (\ref{eq:throughput}) and (\ref{eq:action}) that there are two parameters affecting the state of the environment: i.e., the current channel load G and a node's condition $C$. We first assume that the sensor network is a single agent that interacts with its environment, which includes the channel and itself. This concept is depicted in Fig. \ref{fig:mdp}. We model the problem as an MDP and define the state as
\begin{align}
\begin{split}
S = \times_{1 \leq i\leq m} S_i \times S_u 
\end{split}
\end{align}
where $S \in \pazocal{S}$ is the state of the agent, $\pazocal{S}$ is the set of all states, $S_i$ represents the state of sensor node $i$ and $S_u$ stands for the part of the environment that is uncontrolled by the sensor nodes and corresponds to $G$ in our formulation.  

The transition probabilities of the defined MDP can be formulated as
\begin{align}
P(S_u^{\prime} \given S_u, \times_{1 \leq i\leq n} S_i, K) &= P(S_u^{\prime} \given S_u) \label{eq:uncontrolled}\\
P(S_i^{\prime} \given S_u, \times_{1 \leq i\leq n} S_i, K, F_i) & \propto T_i(K, G) + F_i
\label{eq:state}
\end{align}
where $T_i(\cdot)$ is the individual packet goodput of sensor node $i$ that depends on the current values $K_t$ and $G_t$. Note that we dropped time index $t$ for simplicity of notation.

From (\ref{eq:uncontrolled}), we can see that the transitions of the uncontrollable state $S_u$ are independent of the transmission strategy and the states of individual nodes. Please note here that we assume that the channel probabilistically and stochastically switches states based on the arrival and departure of sensor nodes in the network, changes to noise conditions, etc. Further, from (\ref{eq:state}), we observe that individual transitions of sensor nodes depend on the states and actions of other nodes, channel load, noise conditions and packet throughput. Therefore, transition independence for sensor nodes does not hold. 

The action of the agent $A \in \pazocal{A}$, with $\pazocal{A}$ being the action space,  consists in the joint actions of all the sensor nodes in the network. These actions represent the values of the coefficients $\Lambda_l$ of the probability distribution function in (\ref{eq:distribution}), that is
\begin{equation}
A = A_i \times \cdots \times A_m, \qquad \text{with} \quad A_i = \{\Lambda_{1,i}, \cdots, \Lambda_{d,i}\}
\end{equation}
Recall that $d$ is the maximum number of replicas a sensor node is allowed to send.
\begin{figure}
\centering
\begin{tikzpicture}
\node[draw,rectangle, minimum width =0.8\textwidth, minimum height = 4cm] (env) at (current page.center) {} ;
\node[below right] at (env.north west) {\small Environment};
\node[draw,rectangle, minimum width =0.7\textwidth, minimum height = 2cm] (network) at ([yshift=-0.5cm, xshift=-0.5cm]env) {}; 
\node[below right] at (network.north west) {\small Sensor network};
\node[draw, circle, inner sep=0pt, minimum size =0.7cm, above right] (S1) at ([xshift = 0.7cm, yshift=0.5cm] network.south west) {$s_1$};
\node[draw, circle, inner sep=0pt, minimum size =0.7cm, above ] (S2) at ([yshift = 0.2cm] network.south) {$s_2$};
\node[draw, circle, inner sep=0pt, minimum size =0.7cm, above left] (S3) at ([xshift = -0.7cm, yshift=0.5cm] network.south east) {$s_3$};
\node[above,rectangle, minimum width = 0.1\textwidth, minimum height=0.5cm] (A) at ([yshift=0.5cm]network.north) {a};
\node[draw,rectangle, minimum width=2cm] (inv) at ([xshift=3cm]A) {Medium};
\draw[->] (S1) -- node [text width=1cm, near start, draw=none,rectangle, yshift=0.1cm ] {\small $a_1$} (A);
\draw[->] (S2) -- node [text width=1cm,near start, draw=none,rectangle ] {\small $a_2$} (A);
\draw[->] (S3) -- node [text width=1cm,near start, draw=none,rectangle,above,xshift=0.4cm,yshift=-0.1cm ] {\small $a_3$} (A);
\draw[->] (A) -- (inv);
\draw[->] (inv.east) -| ([xshift=0.2cm]network.north east) to ([xshift=0.2cm]network.east) node [text width=1cm,midway,xshift=0.7cm, draw=none,rectangle] {\small $r,s^{\prime}$} to (network.east);
\end{tikzpicture}
\caption{Intelligent sensor network as an MDP: three sensors ($s_1,s_2,s_3$) transmit their chosen number of replicas $a_1,a_2,a_3$ to the common medium, which responds with a common reward $r$ and the MDP's next state $s^{\prime}$. }
\label{fig:mdp}
\end{figure}
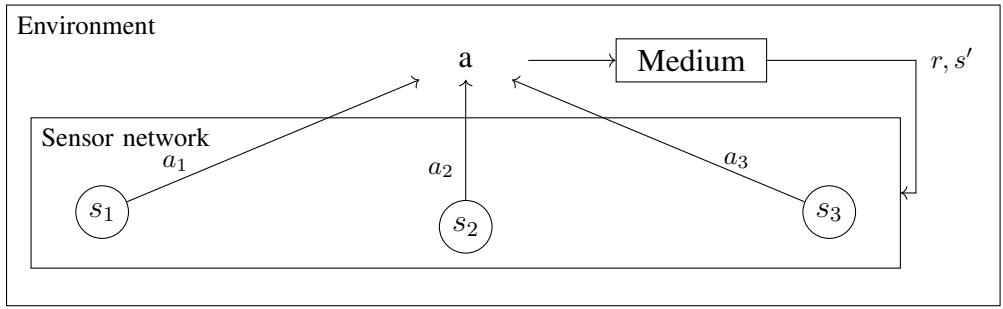
The above MDP formulation, although genuinely modeling the MAC optimization problem, leads to a continuous action space, that scales exponentially with the number of sensor nodes. This renders learning of the optimal action infeasible for large-sized problems. To circumvent this drawback, we redefine the actions as the number of replicas to send. Therefore
\begin{equation}
A = A_i \times \cdots \times A_m, \qquad \text{with} \quad A_i = l, \quad l \in \{1,\cdots, d\}
\end{equation}
During the learning phase the agent finds a deterministic policy $\pi(a|s)$, with $s \in \pazocal{S}$ and $a \in \pazocal{A}$, by choosing the optimal $A_i$ for each sensor node (except for exploratory moves in the learning process where a random action is preferred). After learning has completed, the probability distribution $\Lambda(x)$ is computed using the information of visited state-action pairs. Therefore, upon implementation of our protocol the policy is probabilistic with $\pi(a|s)=\Lambda_a$, where $\Lambda_a$ is the appropriate coefficient in $\Lambda(x)$. This technique allows us to leverage the benefits of maintaining a small action space, while using a stochastic policy. The latter is important in multi-agent scenarios, where existence of an optimal deterministic policy is not guaranteed due to an agent's uncertainty regarding the behavior of other agents \cite{Littman:1994:MGF:3091574.3091594}. 

The choice of the reward function is guided by our aim to design active, self-interested agents attempting to improve overall packet throughput, while lacking access to a global performance measure, i.e., the channel load. We define the immediate reward of an agent as
\begin{equation}
 R^i_t(s, a) = \begin{cases}
       b^i_{t-1} - b^i_{t}, &\text{if} \quad B \rightarrow \infty \\
       -b^i_t , & \text{otherwise } 
       \end{cases}
\end{equation}
where $b^i_t$ is the number of messages in the buffer of sensor node $i$ at time $t$. This reward makes the nodes eager to transmit when their buffers are full, instead of making the decisions purely based on the outcome of the current transmission.

The formulated MDP is episodic with $E$ episodes and $L_E$ learning iterations per episode. At the beginning of an episode each agent can be in a random state $s \in \pazocal{S}$. Experience, in the form of the Q-table and visits to state-action pairs, carries over episodes. Solving the formulated MDP requires finding the optimal transmission policy $\pi$ which is the one that maximized the expected discounted reward starting in state $s$ and then following policy $\pi$. The reward takes into consideration immediate and delayed rewards, and is represented as
\begin{equation}
V_{\pi}(s) = E _{\pi} [\rho_t | S_t = s] = E_{\pi} \Bigg[ \sum_{k=0}^{\infty} \gamma^t R_{t+k+1} | S_t =s\Bigg]
\label{eq:V}
\end{equation}
where $\rho_t$ is the expected return, $R_t$ is the immediate return  and $ 0 \leq \gamma < 1 $ is the discount factor that evaluates the effect of future rewards in the current state (a value of $\gamma$ closer to zero means that the agent is myopic, while when $\gamma$ is close to 1 the agent is farsighted). Equation (\ref{eq:V}) can be rewritten as a Bellman equation \cite{Bellman:2003:DP:862270} 
 \begin{equation}
 V_{\pi} (s) = \sum_a \pi(a|s) \sum_{s^{\prime},r} p(s^{\prime},r|s,a) \big[r+\gamma V_{\pi}(s^{\prime})\big])
 \end{equation}
The main drawback of MDPs is that in many practical scenarios, as in our case, the transition probability $p(s^{\prime},r|s,a)$ and the reward function that generates the reward $R$ are unknown, which makes hard to evaluate policy $\pi$. To this aim, we adopt Q-learning \cite{Watkins1992} that allows to learn from delayed rewards and determine the optimal policy, in absence of the transition probability and reward function. In Q-learning, policies and the value function are represented by a two-dimensional lookup table indexed by state-action pairs $(s,a)$. Formally, for each state $s$ and action $a$, the $Q(s,a)$ value under policy $\pi$, represents the expected discounted reward starting from $s$, taking the action $a$, and thereafter following policy $\pi$. $Q(s,a)$ is defined as follows

\begin{equation}
Q_{\pi}(s,a) = E_{\pi} \Big[ \sum_{k=0}^{\infty} \gamma^t R_{t+k+1} | S_t =s, A_t = a\Big]
\end{equation}

We define the optimal policy as the one that maximizes the expected reward for all states
\begin{equation}
\pi^*(s) = \argmax_{a\in \pazocal{A}} \big( Q^*(s,a)\big) \qquad \text{with} \quad s \in \pazocal{S}
\end{equation}

Bellman's optimality equation for $Q$ allows to define $Q^*$ independently from any specific policy
\begin{align}
Q^*(s,a) &= E\big[R_{t+1} + \gamma \max2_{a^{\prime}} Q^*(S_{t+1}, a^{\prime}) | S_t=s, A_t=a\big] \\
&=\sum_{s^{\prime},r} p(s^{\prime},r|s,a) \big[r+\gamma \max2_{a^{\prime}} Q^*(s^{\prime},a^{\prime})\big])
\end{align}

Using the Q-learning algorithm, a learned action value function Q directly approximates $Q^*$ through value iteration. Correspondingly, the Q-value iterative formula is given by
\begin{equation}
Q(S_t, A_t) = (1-\alpha) Q(S_t,A_t) + \alpha[R_{t} + \gamma \max2_a Q(S_{t+1}, a)]
 \label{eq:qlearn}
\end{equation}     
where $\alpha$ is the learning rate, which determines to what extent newly acquired information overrides old information. The above solution is guaranteed to converge to the optimal solution under the Robbins-Monro conditions: 
\begin{align}
\begin{split}
\sum_t {\alpha_t(s,a)} = \infty \quad \text{and} \quad \sum_t {\alpha_t^2(s,a)} < \infty \qquad
\forall (s,a) \in \pazocal{S} \times \pazocal{A}
\end{split}
\end{align}

As noted earlier, a state consists of all the information necessary for the network to choose the optimal action $a \in \pazocal{A}$. This necessity urges us to encompass in a state $s$ information about  nodes' condition and the uncontrolled state, which includes battery level, buffer size, number of packets to transmit, the channel's noise, load, etc. Clearly, this information cannot be available as it would impose huge communication load, while a MAC protocol should prevent channel congestion and be unintrusive. To alleviate this drawback of MDPs and Q-learning, in the next section we present a novel framework, based on partially-observable MDPs, which has been successfully used in solving problems in resource optimization problems in sensor networks \cite{4480049}.  

\subsection{Dealing with partial observability} 

POMDPs \cite{Kaelbling:1998:PAP:1643275.1643301} acknowledge the inability of an MDP to observe its state, which they remedy by introducing the notion of \textit{observations}. Observations contain information that is relevant but insufficient to describe the actual state on their own. In our case, the network and the sensor nodes cannot observe $S_u = G$, as this requires global knowledge of the environment, which is hard to achieve. We, therefore, constrain observability to information only locally available to the sensor nodes. Following our description in Section \ref{sec:prelim} regarding a sensor node condition $C$, we assume that the only state-related information a node has access to is the number of messages stored in the buffer of each sensor node, that is
\begin{equation}
\label{eq:observations}
\Omega(s) = \Omega_1 \times \cdots \times \Omega_M, \qquad  \text{with} \quad  \Omega_i = b_i
\end{equation}

POMDPs can be optimally solved using the framework of Belief MDPs \cite{Kaelbling:1998:PAP:1643275.1643301}, but this renders learning intractable, as it is performed in continuous state spaces. We adopt a fixed horizon of observations, which is a common approach that, however, has no convergence guarantees. Nevertheless, it has been successfully employed in object tracking problems due to its simplicity and expressive power \cite{AAAI113750}. 

Through the adoption of a fixed history window $w$, the observation tuple of each sensor node is defined as
\begin{equation}
H_t = \{\Omega_{t-w+1}, \cdots, \Omega_{t-1}, \Omega_{t}\}
\label{eq:history}
\end{equation}
and $\Omega$ defined as in (\ref{eq:observations}).

The Belief MDP, whose states correspond to the beliefs over states, is assumed to satisfy the Markov property. Histories of observations serve as an approximation to beliefs, therefore Q-learning can be applied as in the general MDP case.

The distributed nature of the problem has so far been purposely neglected in order to focus on the decision process formulation. Following the observations made in Section \ref{sec:sframework} regarding the need for decentralization, next we proceed by formulating a distributed representation of the problem under the Dec-POMDP framework, introduced in \cite{DBLP:journals/corr/abs-1301-3836}.

\subsection{Dec-POMDP Formulation }\label{sec:dec-pomdp}

Decentralized Partially-observable MDPs offer a powerful framework for designing solutions that take into account partial observability and are controlled in a distributive way. Here, the aim is to design a model-free solution that can help achieve improved overall throughput. The state of the environment includes information about the number of agents and number of slots per time frame, both expressed through $G$. Recall that each agent can only observe its own buffer and thus deduce if its transmission was successful. Fig. \ref{fig:dec} depicts the sensor network as a Dec-POMP.

\theoremstyle{definition}
\begin{definition}{Dec-POMDP}\label{def:decpomdp}
A decentralized partially observable Markov decision process is defined as a tuple $\langle \pazocal{M},\pazocal{S},\pazocal{A},T,R,\Omega,O,w,I\rangle$, where
\begin{itemize}[noitemsep,topsep=0cm]
\item[$\pazocal{M}$] is the set of agents
\item[$\pazocal{S}$] is a finite set of states s in which the environment can be
\item[$\pazocal{A}$] is the finite set of joint actions
\item[$T$] is the transition probability function
\item[$R$] is the immediate reward function
\item[$\Omega$] is the finite set of joint observations
\item[$O$] is the observation probability function
\item[$w$] is the history window
\item[$I$] is the initial state distribution at time t = 0
\end{itemize}
\end{definition}

Definition \ref{def:decpomdp} extends the single-agent POMDP model by considering joint actions and observations. In our case $ A_i \in \{1,2, \cdots, d\}$ , $ \Omega_i \in \{0,1,...,B\}$ and $R_i = - b_i$ is the individual reward agent $i$ observes. As regards the initial distribution $I$, we assume a uniform distribution taking values in the range $[0,1, \cdots, B]$. Note that our algorithm does not need an external, i.e. provided by the environment, common reward function $R$, but agents individually measure their rewards based on their observations.

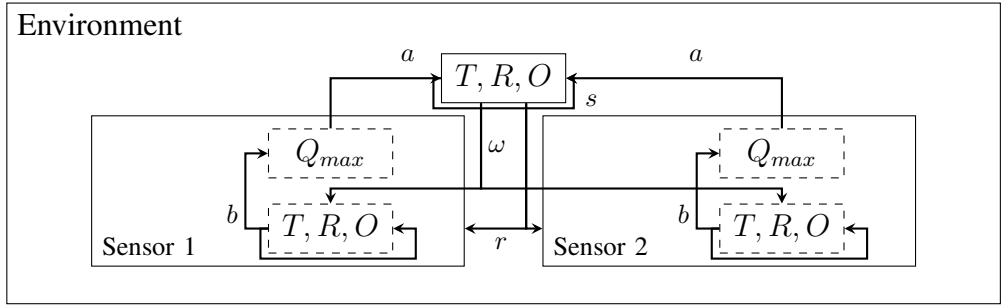
\begin{figure}
\centering
\begin{tikzpicture}
\node[draw,rectangle, minimum width =0.8\textwidth, minimum height = 4cm] (env) at (current page.center) {} ;
\node[below right] at (env.north west) {Environment};
\node[draw,rectangle, minimum width =0.3\textwidth, minimum height = 2cm] (network1) at ([yshift=-0.5cm, xshift=-3cm]env) {}; 
\node[above right] at (network1.south west) {\small Sensor 1};
\node[draw,rectangle, minimum width =0.3\textwidth, minimum height = 2cm] (network2) at ([yshift=-0.5cm, xshift=3cm]env) {}; 
\node[above right] at (network2.south west) {\small Sensor 2};
\node[draw,rectangle,dashed, minimum width =0.1\textwidth, minimum height = 0.5cm] (Q1) at ([yshift=0.5cm, xshift=0.7cm]network1) {$Q_{max}$};
\node[draw,rectangle, dashed,minimum width =0.1\textwidth, minimum height = 0.5cm] (TRO1) at ([yshift=-0.5cm, xshift=0.7cm]network1) {$T,R,O$};
\node[draw,rectangle,dashed, minimum width =0.1\textwidth, minimum height = 0.5cm] (Q2) at ([yshift=0.5cm, xshift=0.7cm]network2) {$Q_{max}$};
\node[draw,rectangle, dashed,minimum width =0.1\textwidth, minimum height = 0.5cm] (TRO2) at ([yshift=-0.5cm, xshift=0.7cm]network2) {$T,R,O$};
\draw[arrow,near start] (TRO1.west) -| node [text width=1cm,  draw=none,rectangle,yshift=0.2cm  ] {\small $b$} ([shift={(-1mm,-4mm)}]TRO1.west)-- ([shift={(3mm,-4mm)}]TRO1.east)|-  (TRO1.east);
\draw[arrow,near start] (TRO2.west) -| node [text width=1cm, draw=none,rectangle,yshift=0.2cm ] {\small $b$} ([shift={(-1mm,-4mm)}]TRO2.west)-- ([shift={(3mm,-4mm)}]TRO2.east)|-  (TRO2.east);
\draw[arrow] (TRO1.west) -|([shift={(-3mm,0mm)}]TRO1.west)-- ([shift={(-3mm,0mm)}]TRO1.west)|- (Q1.west);
\draw[arrow] (TRO2.west) -|([shift={(-3mm,0mm)}]TRO2.west)-- ([shift={(-3mm,0mm)}]TRO2.west)|- (Q2.west);
\node[draw,rectangle, minimum width =0.1\textwidth, minimum height = 0.5cm] (TRO3) at ([yshift=1cm]current page.center) {$T,R,O$};
\draw[arrow] (TRO3.west) -|([shift={(-1mm,-4mm)}]TRO3.west)-- ([shift={(1mm,-4mm)}]TRO3.east)|-  node [text width=1cm,near end, xshift=7mm, draw=none,rectangle,yshift=-0.3cm ] {\small $s$}(TRO3.east);
\draw[arrow] (Q1.north) |-  node [text width=1cm,near end, xshift=7mm, draw=none,rectangle,yshift=0.3cm ] {\small $a$} (TRO3.west);
\draw[arrow] (Q2.north) |-  node [text width=1cm,near end, xshift=7mm, draw=none,rectangle,yshift=0.3cm ] {\small $a$} (TRO3.east);
\draw[arrow] ([xshift=-0.3cm]TRO3.south) |- node [text width=1cm,near start, xshift=6mm, draw=none,rectangle ] {\small $\omega$} ([shift={(1cm,2mm)}]TRO1.north)-- ([shift={(0mm,2mm)}]TRO1.north) --  (TRO1.north);
\draw[arrow] ([xshift=-0.3cm]TRO3.south) |-([shift={(-1cm,2mm)}]TRO2.north)-- ([shift={(0mm,2mm)}]TRO2.north) -- (TRO2.north);
\draw[arrow] ([xshift=0.3cm]TRO3.south) |-  ([shift={(-0.2cm,-5mm)}]network2.west)--  ([yshift=-5mm]network2.west);
\draw[arrow] ([xshift=0.3cm]TRO3.south) |- ([shift={(0.2cm,-5mm)}]network1.east) -- node [text width=1cm,midway, xshift=8mm,yshift=-2mm, draw=none,rectangle] {\small $r$} ([yshift=-5mm]network1.east);
\end{tikzpicture}
\caption{Intelligent sensor network expressed as Dec-POMDP}
\label{fig:dec}
\end{figure}

As we mentioned in Section \ref{sec:intro}, the decentralization property of the POMDP framework changes the nature of the problem to NEXP-complete, a class of problems too complicated to provide any real-time solution. Nevertheless, the theoretical properties of this family of problems have been studied and efficient algorithms have been developed in \cite{Kaelbling:1998:PAP:1643275.1643301}. For example, the Witness algorithm is introduced in \cite{Kaelbling:1998:PAP:1643275.1643301} as a polynomial time alternative to value iteration in policy trees. As these algorithms suffer from extreme memory requirements due to the continuous nature of the problem, locality of interaction has been leveraged in the Networked Distributed POMDP setting \cite{Nair:2005:NDP:1619332.1619356}, where LID-JESP and GOA are introduced for planning in Dec-POMDPs. Contrarily to the above, we will use model-free remedies to circumvent the inherent intractability, an approach that will benefit from lower complexity, both in terms of computation and time.

To learn the optimal policy using a model-free approach one can apply simple single-agent Q-learning. This is performed as follows
\begin{equation}\label{eq:qlearn}
Q(H_t, A_t) = (1-\alpha) Q(H_t,A_t) + a[R_{t} + \gamma \max2_a Q(H_{t+1}, a)]
\end{equation}
Although this approach leads to an optimal policy, it is inappropriate in the Dec-POMDP framework as adopting a centralized point of control creates a large state space and demands global access to information. In \cite{Claus:1998:DRL:295240.295800} independent learning, in which each agent learns its own Q-value function ignoring other agents' actions and observations, is studied. By ignoring the effect of interaction among agents this approach may converge to local optimal policies or oscillate. Nevertheless, independent learning offers a distributed, tractable solution that has proven adequate in relevant applications \cite{AAAI113750}. Motivated by the encouraging results in \cite{Claus:1998:DRL:295240.295800}, we formulate the problem as a population of agents which make decisions independently of each other on how to handle common resources in order to maximize social welfare, i.e., the overall throughput. Our adoption of the powerful framework of Dec-POMDP is justified by the realistic nature of MAC protocol design, as its success will depend on the achievement of low complexity. The work in \cite{Lesser:2003:DSN:940763}  strained the importance of realistic modeling  in networking applications, as specific characteristics have a significant, algorithm-specific impact on the solution. In Section \ref{sec:experiments}, we will experimentally investigate the performance of independent learning under various learning settings in order to draw qualitative conclusions about the appropriate behaviors of sensor agents and design a MAC protocol that surpasses the performance of vanilla IRSA.

\subsection{Virtual experience}

%
%
Q-learning is a model-free learning approach, however due to its conceptual simplicity proves to be inefficient for real-time applications, as extensive interaction with the environment is required. Leveraging past experience is a technique that has successfully been used in demanding RL tasks due to its effectiveness and its respect to the structural properties of Q-learning \cite{MnihKSGAWR13}. Key intuition behind it, is that an agent can update the Q-values of states it has previously visited. These batch updates can significantly decrease convergence time, provided that the agent avoids acting on outdated information. A related notion is that of \textit{virtual experience} \cite{5986747, DBLP:journals/corr/ThomosKFS14}, where an agent ``imagines" state visits  instead of ``remembering" them. The work in \cite{5986747} separated the effect of the environment into ``known'' and ``unknown'' dynamics and  introduced the notion of virtual experience in their attempt to extrapolate experience of actual rewards to states that do not affect the unknown dynamics and are, therefore, equivalent in the light of new information. Virtual experience was applied to post-decision states, and not to actual states. Next, we will proceed by formulating virtual experience in the observational histories of our own learning setting.

As defined in (\ref{eq:history}), an agent's history of observations is a tuple of past buffer states. Based on this information, an agent chooses the preferred number of replicas to send. The unknown environment dynamics in this case include the arrival and collision model, take place after the selection of replica's number, and determine the reward the agent experiences as well as the next observation $\omega$. Although agent's $i$ observation vector $h_i$ is essential for determining the optimal action, we should point out that the unknown dynamics do not directly depend on $h_i$. In particular, if the observation tuple is $H_t^i = \{b^i_{t-w+1}, \cdots, b^i_{t-1},b^i_{t}\}$, then the collision model cannot discern any difference in states of the following form
\begin{align}
H^{i^{\prime}}_t &= \{b^{i^{\prime}} _{t-w+1},\cdots, b^{i^{\prime}}_{t-1}, b^{i^{\prime}}_{t}\} \nonumber\\
    &= \{ b^{i^{\prime}}_{t-w+1}, \cdots,  b^{i^{\prime}}_{t-2} - c_{t-1}^i,  b^{i^{\prime}}_{t-1} - c_{t}^i\},  \\
\text{with} \qquad c^i_t &= b^i_{t-1} - b^i_{t}
\end{align}
where $c^i_t$ is the difference in observations between two consecutive states.

Virtual experience can be viewed as applying the transformation formulated in (\ref{eq:transform}) on visited states and then updating all states that have the same representation. We call $\widetilde{h}$ a virtual state, as it is neither visited nor directly used in the Q-learning update, but serves as an intermediate state in order to acknowledge states equivalent towards the unknown environment dynamics. We illustrate this in Fig. \ref{fig:virtual}.
\begin{equation}
\label{eq:transform}
H_t= \{ b_{t-w+1}, \cdots,  b_{t-2} - c_{t-1},  b_{t-1} - c_{t}\} \xrightarrow{\pazocal{T}} \widetilde{H} = \{c_{t-w+2}, \cdots,  c_t\} 
\end{equation}
The reason for the above formulation is that collisions should intuitively depend on the relative buffer states $\widetilde{h}$, as they determine the channel congestion. The actual values are useful in shaping the eagerness of agents to transmit data. Formally and according to \cite{MastroThesis} a pair $(\tilde{s}, \tilde{a})$ is equivalent to a pair $(s,a)$ if $p(s^{\prime} | s,a) = p(s^{\prime} | \tilde{s}, \tilde{a}), \ \forall s^{\prime} \in S$ and $R(\tilde{s}, \tilde{a})$ can be derived from $R(s,a)$. Following the above observation for each move of an agent a batch update on all pairs $(s_j,a_j)$ with $T(s_j) = \widetilde{h}$ and $a_j = a$  will be performed. Note that we cannot extrapolate experience to states with different actions, as the collision dynamics depend on the action performed.
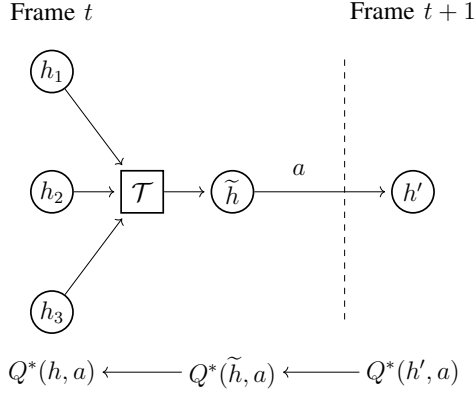
\begin{figure}
\centering
\scalebox{0.8}{
\begin{tikzpicture}[every node/.style={draw,thick,circle,inner sep=0pt, minimum size =0.7cm}]
\node[circle,] (virt) at (current page.center) {$\widetilde{h} $} ;
\node[circle] (s1) at ([xshift=-3cm,yshift=2cm]virt) {$h_1$} ;
\node[circle] (s2) at ([xshift=-3cm]virt) {$h_2$} ;
\node[circle] (s3) at ([xshift=-3cm,yshift=-2cm]virt) {$h_3$} ;
\node[circle] (s') at ([xshift=3cm]virt) {$h^{\prime}$} ;
\node[draw,rectangle] (T1) at ($(s2)!0.5!(virt)$) {$\pazocal{T}$};
\draw[shorten >=0.1cm,->] (s2) -- (T1);
\draw[shorten >=0.1cm,->] (s1) -- (T1);
\draw[shorten >=0.1cm,->] (s3) -- (T1) ;
\draw[shorten >=0.1cm,->] (T1) -- (virt);
\draw[shorten >=0.1cm,->] (virt) -- node[text width=1cm, midway,above, draw=none,rectangle ]  {$a$}  (s');
\node[rectangle,draw=none] (V1) at([yshift=-1cm]s3) {$Q^{*}(h, a)$};
\node[rectangle,draw=none] (V2) at([yshift=-1cm,xshift=3cm]s3) {$Q^{*}(\widetilde{h}, a)$};
\node[rectangle,draw=none] (V3) at([yshift=-1cm,xshift=6cm]s3) {$Q^{*}(h^{\prime}, a)$};
\draw[shorten >=0.1cm,shorten <=0.1cm,->] (V3) -- (V2);
\draw[shorten >=0.1cm,shorten <=0.1cm,->] (V2) -- (V1);
\node[rectangle,draw=none] (T1) at([yshift=1cm]s1) {Frame $t$};
\node[rectangle,draw=none] (T2) at([yshift=1cm,xshift=6cm]s1) {Frame $t+1$};
\draw[dashed] ([shift={(1.5cm,2.2cm)}]virt.east) -- ([shift={(1.5cm,-2.2cm)}]virt.east);
\end{tikzpicture}}
\caption{Virtual experience: $h_1$, $h_2$ and $h_3$ are different states of the Dec-POMDP that are mapped to the same $\widetilde{h}$ through the transformation $\pazocal{T}$. After common action $\alpha$ is performed and the next $h^{\prime}$ is observed, the Q-table is updated for all of them. }
\label{fig:virtual}
\end{figure}
Equipping Q-learning with virtual experience increases computational complexity, as instead of updating one entry of the Q-table in each learning iteration, all $(h,a)$ pairs with the same $\widetilde{h}$ are updated. This complexity increase is equal to the number of those pairs, which we denote by $|\widetilde{H}|$ and can be bounded as
\begin{align}
\label{eq:h}
&0 \leq |\widetilde{\pazocal{H}}| \leq \min \{B+1-B_{\text{min}}, B_{\text{max}}\} & \text{where}\\
B_{\text{min}}&= \argmin_b  \{ b-\sum_{t=w-1}^{\tau}{c_t} \geq 0 \}  & \text{and} \\
B_{\text{max}}&= \argmax_b \{ b-\sum_{t=w-1}^{\tau}{c_t} \leq B \} & \forall \tau \in [0,w-2)
\end{align}
$B_{\text{min}}$ and $B_{\text{max}}$ are used to avoid considering virtual states with numbers of packets in their buffers that are either negative or exceed the maximum capacity $B$.

The conception of virtual experience in \cite{MastroThesis} was not accompanied by its theoretical properties regarding convergence time, we therefore conclude this section with some remarks on the effect of virtual experience on it. Inspired by the work in \cite{Even-Dar:2004:LRQ:1005332.1005333}, where convergence time of Q-learning was studied in relation to its parameters, e.g., the learning rate and discount factor, and lower bounds were computed for synchronous and asynchronous learning using polynomial and linear learning rates, we study how virtual experience affects convergence time and derive a similar bound. We limit ourselves to asynchronous learning using a polynomially decreasing learning rate, as is the current case, and extend it by considering multiple updates in each iteration.

We first study how virtual experience affects coverage time $L$, i.e., the learning iterations necessary to visit all state action pairs at least once and then proceed to bounding convergence time. Our remarks will be based on Lemma 33 from \cite{Even-Dar:2004:LRQ:1005332.1005333}.

%
%
%

\begin{lemma}
Assume that $P$ is the probability of visiting all state action pairs in an interval $k$, where an interval corresponds to a time period of $L$ iterations. Then, using virtual experience, the probability $P^{\varv}$ of visiting all state-action pairs $P$ in an interval $k$ is $|\widetilde{\pazocal{H}}|P$. 
\end{lemma}

\begin{proof}
The probability $P$ can also be interpreted as the percentage of unique pairs visited, i.e., $P=L_{u}/(|\pazocal{S}|\times |\pazocal{A}|)$, where $L_u$ is the number of iterations where the pair was visited for the first time and the denominator represents the size of the state-action space. We assume that states are sampled with replacement from an i.i.d. probability distribution. As noted earlier, virtual experience increases the number of states updated in a learning iteration by $|\widetilde{\pazocal{H}}|$, with $|\widetilde{\pazocal{H}}|$ defined in (\ref{eq:h}). It follows then that $L_u^{\varv} =|\widetilde{\pazocal{H}}| L_u$, where $L_u^{\varv}$ is the number of iterations where the visited pair was unique using virtual experience. Thus, $P^{\varv} = |\widetilde{\pazocal{H}}| P$.
\end{proof}

\begin{lemma}
\label{lemma3}
Assume that from any start state we visit all state-action pairs with probability $|\widetilde{\pazocal{H}}|P$ in $L$ steps. Then with probability $1-\delta$ from any initial state we visit all state-action pairs in $L\frac{\log_2(\delta)}{\log_2(1-|\widetilde{\pazocal{H}}|P)}$ steps for a learning period of length $[\frac{\log_2(\delta)}{\log_2(1-|\widetilde{\pazocal{H}}|P)}]$.
\end{lemma}

\begin{proof}
The probability of not visiting all state-action pairs in $k$ consecutive intervals is $(1-|\widetilde{H}|P)^k$. If we define $k$ as $\log_{1-|\widetilde{\pazocal{H}}|P} (\delta)$, then this probability equals $\delta$ and $L\log_{1-|\widetilde{\pazocal{H}}|P}(\delta) = L\frac{\log_2 (\delta)}{\log_2 (1-|\widetilde{\pazocal{H}}|P)}$ steps will be necessary to visit all state-action pairs.
\end{proof}

\begin{corollary}\label{cor1}
Virtual experience alters coverage time $L$ by a factor of $\frac{\log_2(1-P)}{\log_2(1-|\widetilde{\pazocal{H}}|P)}$.
\end{corollary}

According to \cite{Even-Dar:2004:LRQ:1005332.1005333}, convergence time depends on the covering time based on the following theorem.
\begin{theorem}\label{theo1}
Let $Q_t$ be the value of the asynchronous Q-learning algorithm using polynomial learning rate at time $\tau$. Then, with probability at least $1-\delta$, we have $\norm{Q_t-Q^*} \leq \epsilon$, given that
\begin{equation*}
\tau= \Omega(L^{3+1/\phi} + L^{1/(1-\phi)})
\end{equation*}
\end{theorem}

where $\phi$ is a parameter that determines how fast the learning rate converges to zero, i.e, $\alpha = 1/t^\phi$.
\begin{proof}
The proof is identical with Theorem 4 in \cite{Even-Dar:2004:LRQ:1005332.1005333}.
\end{proof}

Combining Theorem \ref{theo1} and Lemma \ref{lemma3} we conclude that the lower bound for convergence time is reduced to $\Big(L\frac{\log_2 (\delta)}{\log_2 (1-|\widetilde{\pazocal{H}}|P)}\Big)^{3+1/\phi} + \Big(L\frac{\log_2 (\delta)}{\log_2 (1-|\widetilde{\pazocal{H}}|P)}\Big)^{1/(1-\phi)}$ .
\subsection{Computational complexity}\label{sec:complexity}
The proposed protocol is a computationally attractive alternative to transmission strategies that are based on finite length analysis \cite{7247494}, which has exponential complexity. In our framework, at each learning iteration an agent has to choose its transmission strategy and then update its local Q-table. In contrast to the work in \cite{Toni2018}, in the proposed scheme the action space is discrete and increases linearly with  $d$. The size of the observation space, which coincides with the size of the Q-table, is $(B+1)^w$, where $B$ is the size of sensor nodes' buffer and $w$ is the history window. The observation space  scales exponentially with $w$ and linearly with $B$. Finally, the complexity associated with the number of agents is $O(1)$, as each agent learns independently.

\section{Simulation results}\label{sec:experiments}

This section begins with a performance comparison of the proposed Dec-RL IRSA protocol and vanilla IRSA. It subsequently studies the effect of different learning schemes on the performance of independent learning with the two-fold goal of drawing conclusions about the behavior of agents and providing a guideline for configuring system parameters to determine the optimal strategy. Finally, we evaluate the proposed scheme advanced with the virtual experience concept to show the reduced convergence time.

\subsection{Simulation Setup}

The following experiments are performed on a toy network with frames of size 10 and channel loads $G \in [0.1,\cdots,1 ]$, which remain constant throughout the learning and simulation of communication time. Unless stated otherwise, performance is averaged over 1000 Monte Carlo trials, the number of sensor nodes is determined by $M=G\cdot N$,  learning requires 1500 iterations and confidence intervals are calculated based on 20 independent experiments with $97.5\%$ confidence level. As regards configuration of learning, we experimentally validated that e-greedy exploration with a constant exploration rate $e=0.05$, a decreasing learning rate following $\alpha = 1.111 \cdot 0.9^{iter}$ formula, where $iter$ is the number of times the current state-action pair has been visited, and a constant discount factor $\gamma=0.98$ offer the optimal policy. As a baseline method for our comparisons we use IRSA with $\Lambda(x) = 0.25x^2 + 0.60x^3 + 0.15x^8$, which was experimentally evaluated in the work of \cite{5668922} and proved superior to other commonly used distributions derived in \cite{Storn1997}.
%
%
%

\subsection{Protocol Comparison}

Based on the observations of the work in \cite{Lesser:2003:DSN:940763}, a protocol orchestrating a multi-agent system should be examined in the regard of the following properties: completeness, i.e., its ability to find the optimal solution, if any, rate of convergence, complexity and scalability. Of these, completeness is a requirement often dropped in real-time, non-stationary environments, as convergence to a good solution is more valued than exhausting one's resources, i.e., CPU power, time and memory,  in the vain pursuit of the optimal one. As regards scalability, our method is invariant to the number of agents due to independent learning, while the complexity  scales exponentially with the size of the observation history. Nevertheless, as we show later, our scheme gets most of the benefits from the history consideration by adopting a short history window. Hence, complexity is not an issue for our solution.

Fig. \ref{fig:statistical} performs a statistical analysis on the performances of the two protocols under consideration by presenting confidence intervals. From this figure, it is obvious that Dec-RL IRSA is superior to vanilla IRSA in all cases with the difference gap becoming wider for channel loads above $0.6$. We also observe that performance has higher variations in high channel loads. Fig. \ref{fig:convergence} illustrates convergence time for independent learning in different channel loads. From this figure, we can see that convergence is guaranteed and is fast for low channel loads. For $G=0.2$ only four learning iterations are necessary, while for $G=0.4$ seven iterations are needed. In the case of high channel loads Dec-RL IRSA fails to transmit messages faster than their arrival rate, the node's buffer thus saturates fast to $r_i = -B$ for $G=1$ and tends to saturate at the end of the episode for $G=0.8$. Based on this observation, we design a mechanism for agents to detect ``\textit{bad}'' episodes and reset the POMDP to an arbitrary state. We classify an episode as ``\textit{bad}'' if the rewards deteriorate for three consecutive iterations.

\begin{figure}[t]
\centering
\begin{minipage}[t]{.45\textwidth}
\includegraphics[scale = 0.4]{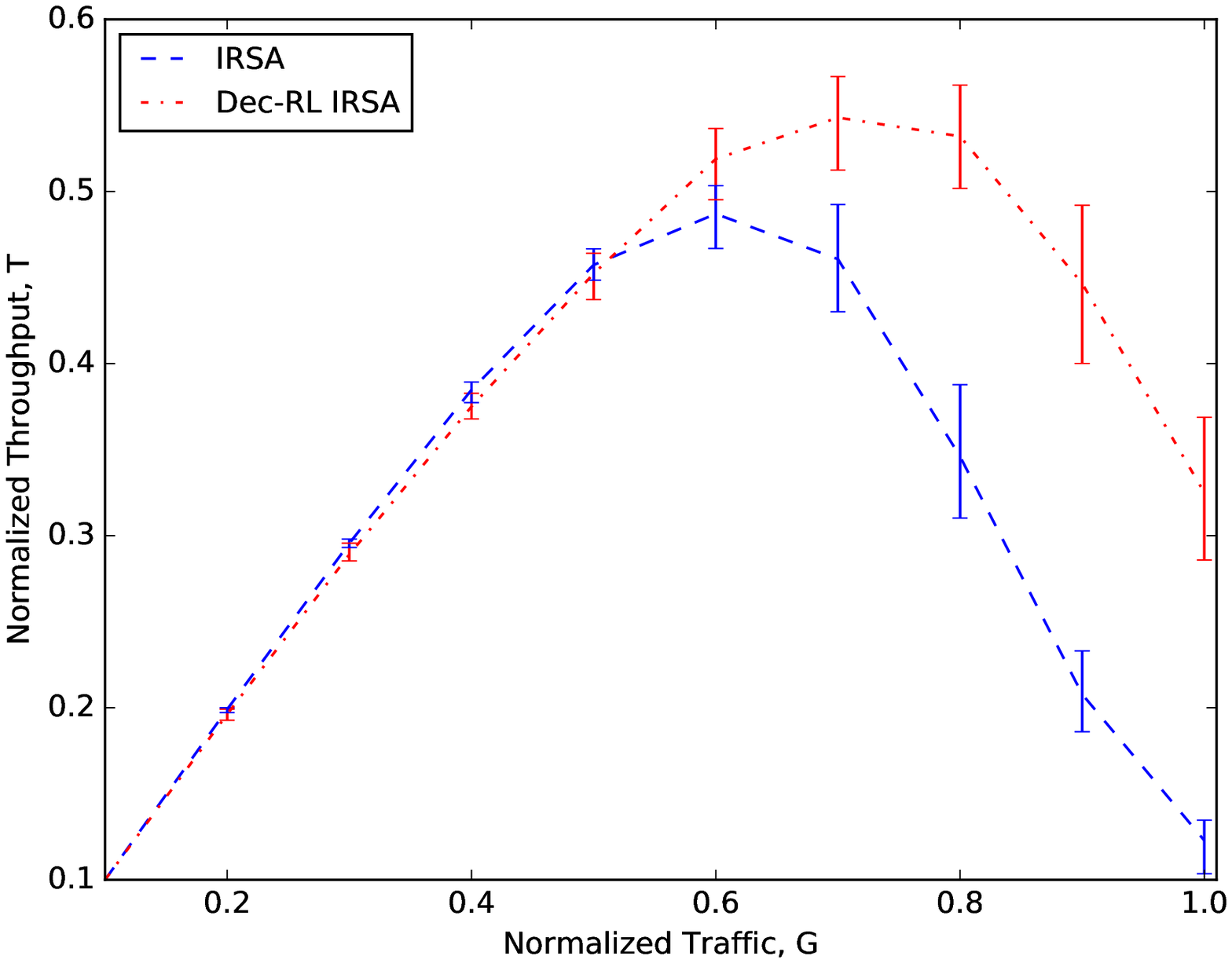}
\caption{Achieved throughput comparison of IRSA and Dec-RL IRSA on a toy network for varying channel loads. Confidence intervals were calculated on 20 independent experiments, each one with 250 Monte Carlo trials.}
\label{fig:statistical}
\end{minipage}\hfill
\begin{minipage}[t]{.45\textwidth}
\includegraphics[scale = 0.4]{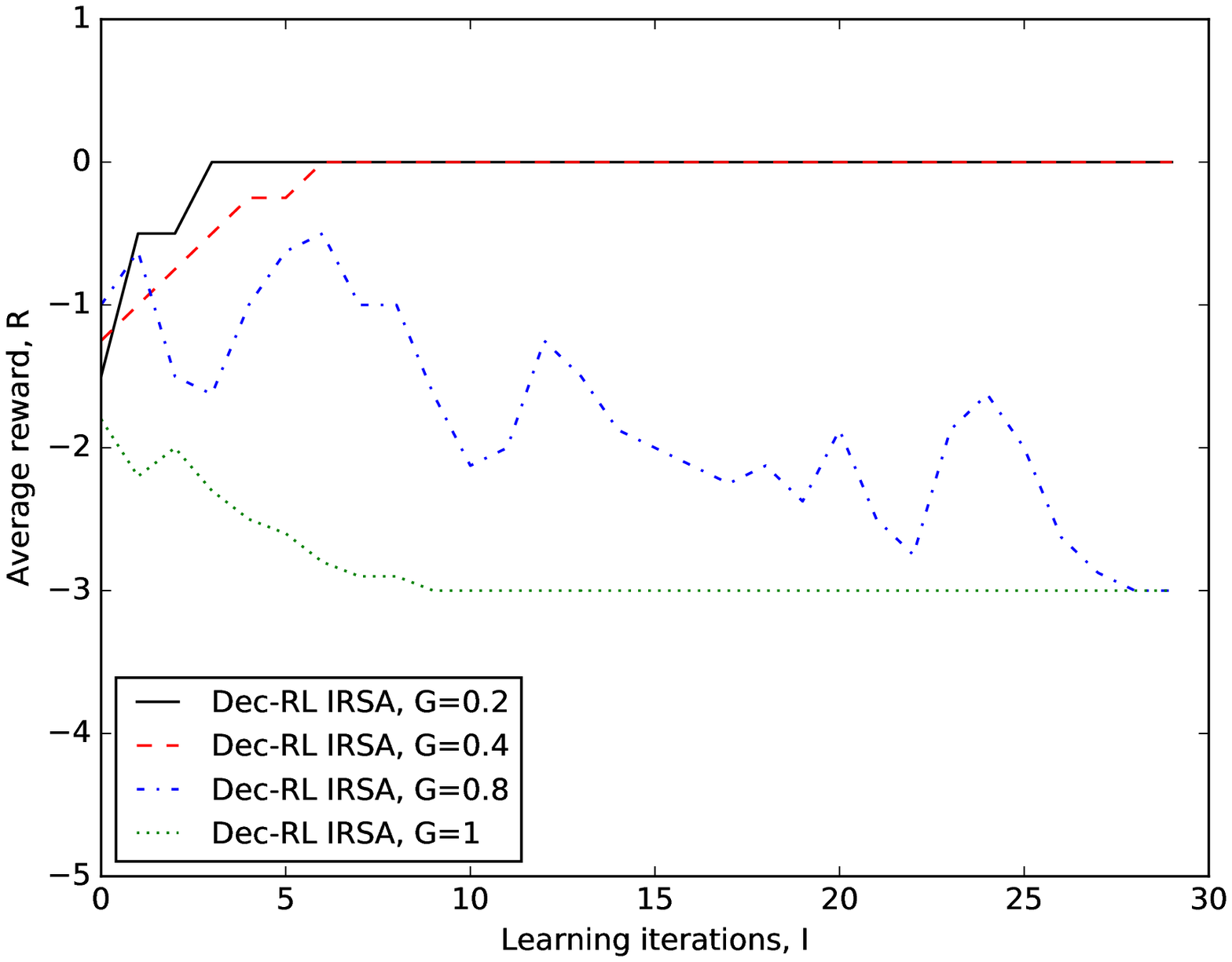}
\caption{Average rewards of Dec-RL IRSA for different channel loads $G$ and the $25^{th}$ episode of its learning time in a run of  50 episodes. }
\label{fig:convergence}
\end{minipage}
\end{figure}

Fig. \ref{fig:scale} illustrates how Dec-RL IRSA and vanilla IRSA achieved throughput changes with different frame sizes ($G=0.6, 0.8, 1$). As regards scalability of Dec-RL IRSA, it appears robust and its performance increases with bigger frame sizes. This can be attributed to the fact that learning is more effective in more complex networks, where collisions occur more often, thus, learning to avoid other agents has a more profound impact on the overall throughput. Vanilla IRSA also improves its performance for increased frame sizes, as it provingly works better in asymptotic settings. This is attributed to the fact that the probability distribution $\Lambda(x)$ is computed using asymptotic analysis and is therefore closer to optimal for frames that exceed 200 time slots. Nevertheless, the performance gap  of the observed throughput of Dec-RL IRSA compared with vanilla IRSA remains high in heavy channel loads ($G=1$), due to the waterfall effect of vanilla IRSA. To conclude scalability analysis, the slight superiority of vanilla IRSA manifested for low $G$ in asymptotic settings is irrelevant to practical scenarios, as the assumption of very large frame size $N$ leads to inefficient implementations, in particular in sensor and IoT networks, that require a complex receiver and introduce delay.

\subsection{Effect of state space size}

The size of the state space, i.e., the number of possible states for an agent, depends on the length of the history of observations, as well as the maximum value of the observations, which is equal to $B+1$ , the size of the buffers of agents. Increasing $B$ has a two-fold effect. Firstly, it increases the size of the state space, thus making learning harder due to the need for longer exploration. Secondly, it dilates the range of rewards, thus agents are made more eager to transmit. Assuming buffer sizes of constant size, constrained by characteristics of the sensor nodes, one anticipates to improve performance of learning by increasing the history window, as that will lead to better approximation of actual states. Nevertheless, letting memory constraints aside, this will lead to an exponential increase of the world size leading either to intractable problems or high time requirements. Thus, it is crucial to determine the minimum amount of information necessary for agents to derive efficient policies. Note that for the sake of a fair comparison learning iterations were also increased to 3000 for increased history window and buffer size. Fig. \ref{fig:buffer} demonstrates that using a value of $B=1$, i.e., only one packet is kept in the buffer, leads to lower throughput for channel loads above $0.6$, as agents are not made eager enough to transmit. On the other hand, increased buffer size improves the perceived throughput for loads above $0.8$, but it slightly degrades it for the rest.

Regarding history size, Fig. \ref{fig:history} reveals that the effect of increased world size is more profound. This results from the fact that, according to Section \ref{sec:complexity}, size scales exponentially with $w$ and linearly with $B$. We observe that by decreasing the window to $w=2$, a severe degradation in performance is observed, suggesting that the information provided to the agents through the observation tuples is not substantial. Increasing the learning iterations for $w=8$ has a counterintuitive effect, as performance is degraded, whereas we would expect that an increased world size would benefit from larger training times. In this case, $800$ learning iterations perform optimally, so we can assume that by equipping agents with larger memory leads to learning of better actions. Still, considering the current parameterization, $w=4$ is the best performing choice.
\begin{figure}[t]
\centering
\begin{minipage}[t]{.45\textwidth}
\includegraphics[scale = 0.4]{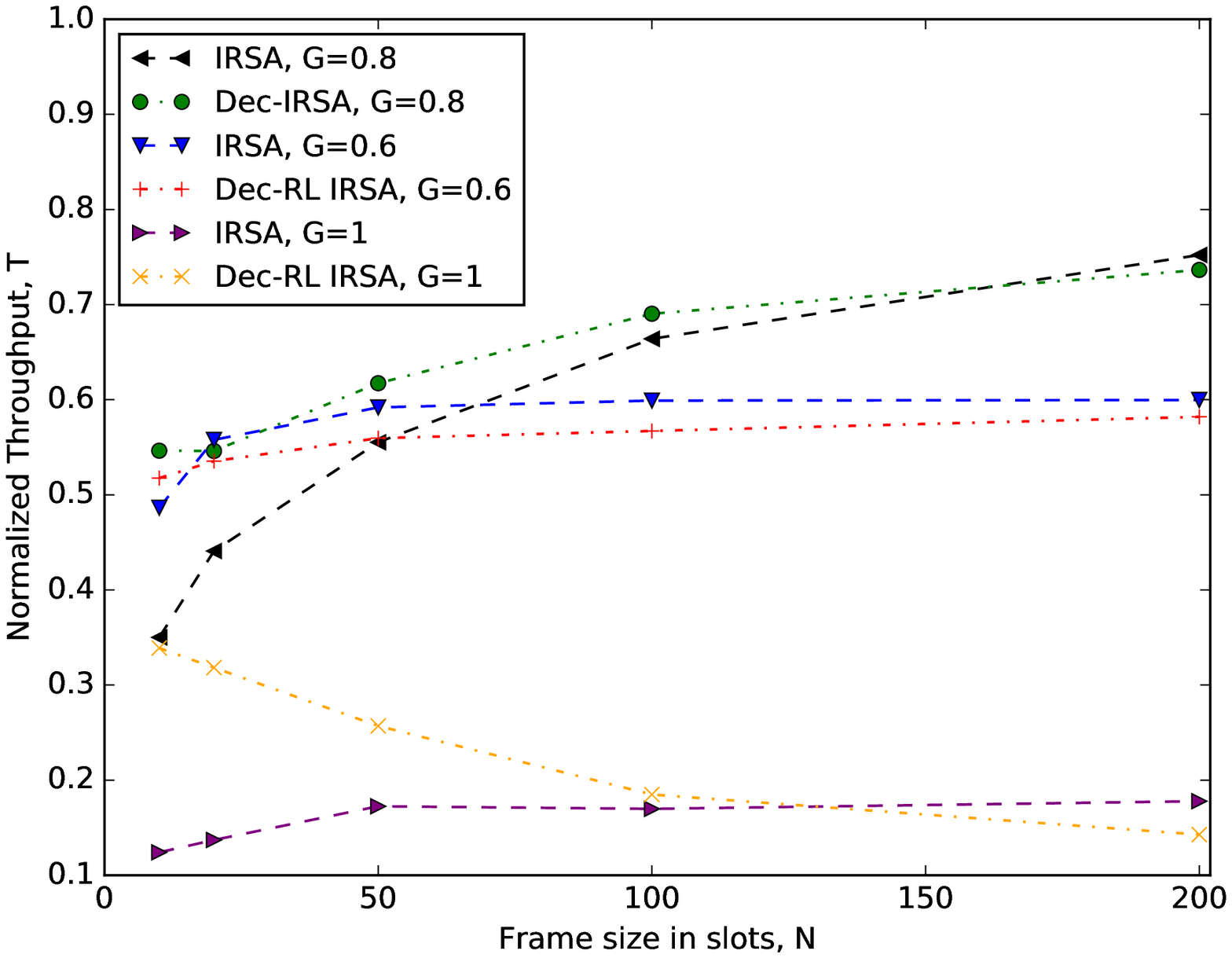}
\caption{Achieved throughput comparison of IRSA and Dec-IRSA for varying frame sizes $N$ and channel traffic $G=0.6$ and $G=0.8$.}
\label{fig:scale}
\end{minipage}\hfill
\begin{minipage}[t]{.45\textwidth}
\includegraphics[scale = 0.4]{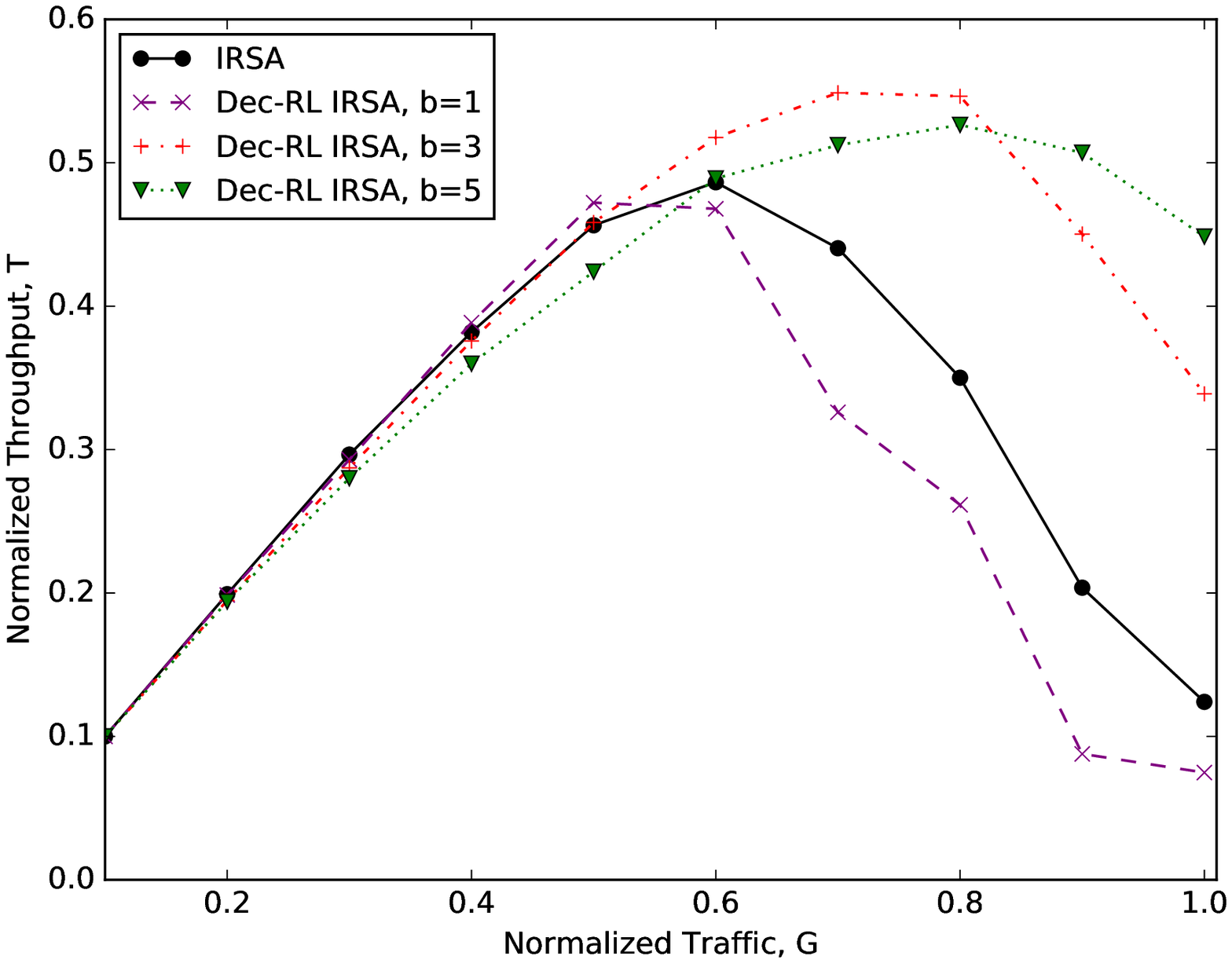}
\caption{Comparison of achieved throughput for different buffer sizes of sensors.}
\label{fig:buffer}
\end{minipage}
\end{figure}
\subsection{Virtual experience}\label{sec:expvirt} 
Virtual experience was introduced to reduce convergence time, which we experimentally measure using the \textit{weighted percent error metric} and $\epsilon$-convergence time, similar to the work in \cite{5986747}. Fig. \ref{fig:virtual} shows how throughput varies for different number of learning iterations and suggests that, using virtual experience, the optimal number of iterations was reduced from 1500 to 500. Fig. \ref{fig:statistical_virtual} performs a statistical analysis on $\epsilon$-convergence time for different channel loads using a $95\%$ confidence level on 40 independent experiments and $\epsilon=0.5$. We observe that convergence is fast for low loads regardless of the use of virtual experience. For $G \geq 0.6$, however, we observe that virtual experience exhibits an improvement of around $80\%$, which can be attributed to increasing the number of batch updates by a factor of $\widetilde{\pazocal{H}}$. Also, vanilla Dec-IRSA usually fails to converge for high channel loads, although throughput remains close to optimal. This observation suggests that, in this case, there are different policies that lead to optimal behavior, so vanilla Dec-RL IRSA is less biased to the optimal one. Note that the degradation in performance with increasing learning iterations, observed in Fig. \ref{fig:virtual} and manifested at around 1500 iterations for vanilla Dec-RL IRSA and 500 using virtual experience, is attributed to over-training.  
%
%

\subsection{Waterfall effect} 

The performance of IRSA has been proven to be governed by a stability condition \cite{5668922} which leads to a waterfall effect similar to the one observed in the decoding of LDPC codes \cite{910578}. From a learning perspective, this profoundly changes the nature of the problem and thus the learning objective. As described in Section \ref{sec:dec-pomdp}, the problem is one of agents competing for a pool of common resources.This formulation resembles the El Farol bar problem, a well-studied scenario in the reinforcement learning literature, but this description is not rich enough to illustrate the learning objectives of individual agents. In the realm of low channel loads ($G < G^*$), where resources are abundant, agents must learn to coordinate their actions, as there is a number of replicas to transmit that optimizes packet throughput. Note that for low channel loads ($G \leq 0.5$) even a random strategy is appropriate, so learning is of no practical interest. In the realm of high channel loads ($G \geq G^*$) however we can acknowledge the task as a Dispersion game \cite{GrenagerAAAI02}, where agents need to  cooperate in order to avoid congesting the channel by exploiting it in different time frames. Different problem nature  urges for different learning behavior, thus we expect that parameterization of learning should vary with $G$. Fig. \ref{fig:waterfall} illustrates the performance of three different parameterizations, each one optimal for a different range of values for $G$. The random strategy was implemented by sampling the number of replicas $l$ uniformly from $\{1,\cdots,d\}$ at each node's transmission. Note that $G^∗$ and $G^{\text{low}}$ stand for the threshold below which the probability of unsuccessful transmission is negligible and a random strategy is optimal, respectively. We observe that by optimizing the parameters for a particular range of $G$ values, we obtain significant gains in the region of interest ($G > 0.6$).

\begin{figure}[t]
\centering
\begin{minipage}[t]{.45\textwidth}
\includegraphics[scale=0.4]{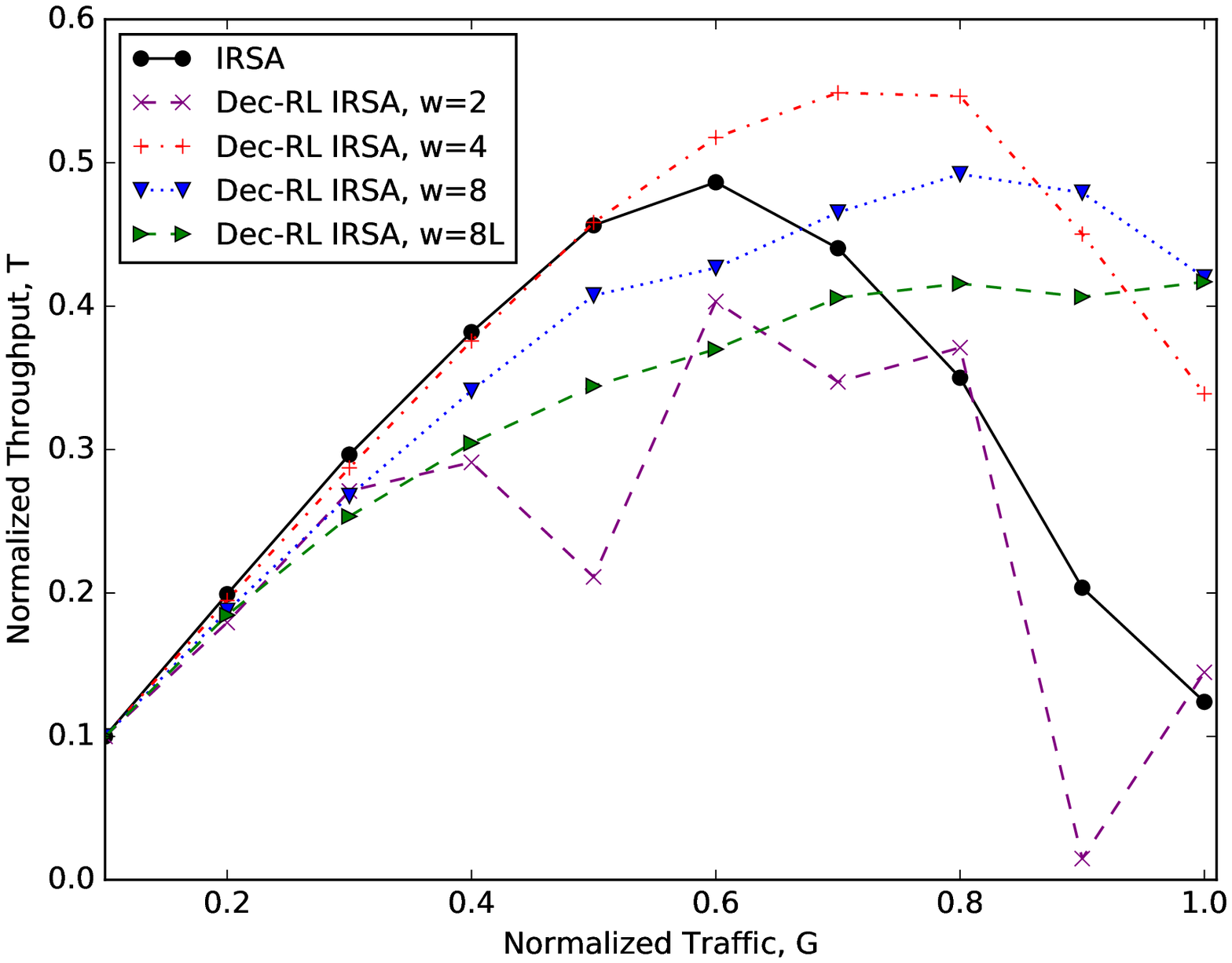}
\caption{Comparison of achieved throughput for different history windows.}
\vspace{-0.8cm}
\label{fig:history}
\end{minipage} \hfill
\begin{minipage}[t]{.45\textwidth}
\includegraphics[scale = 0.4]{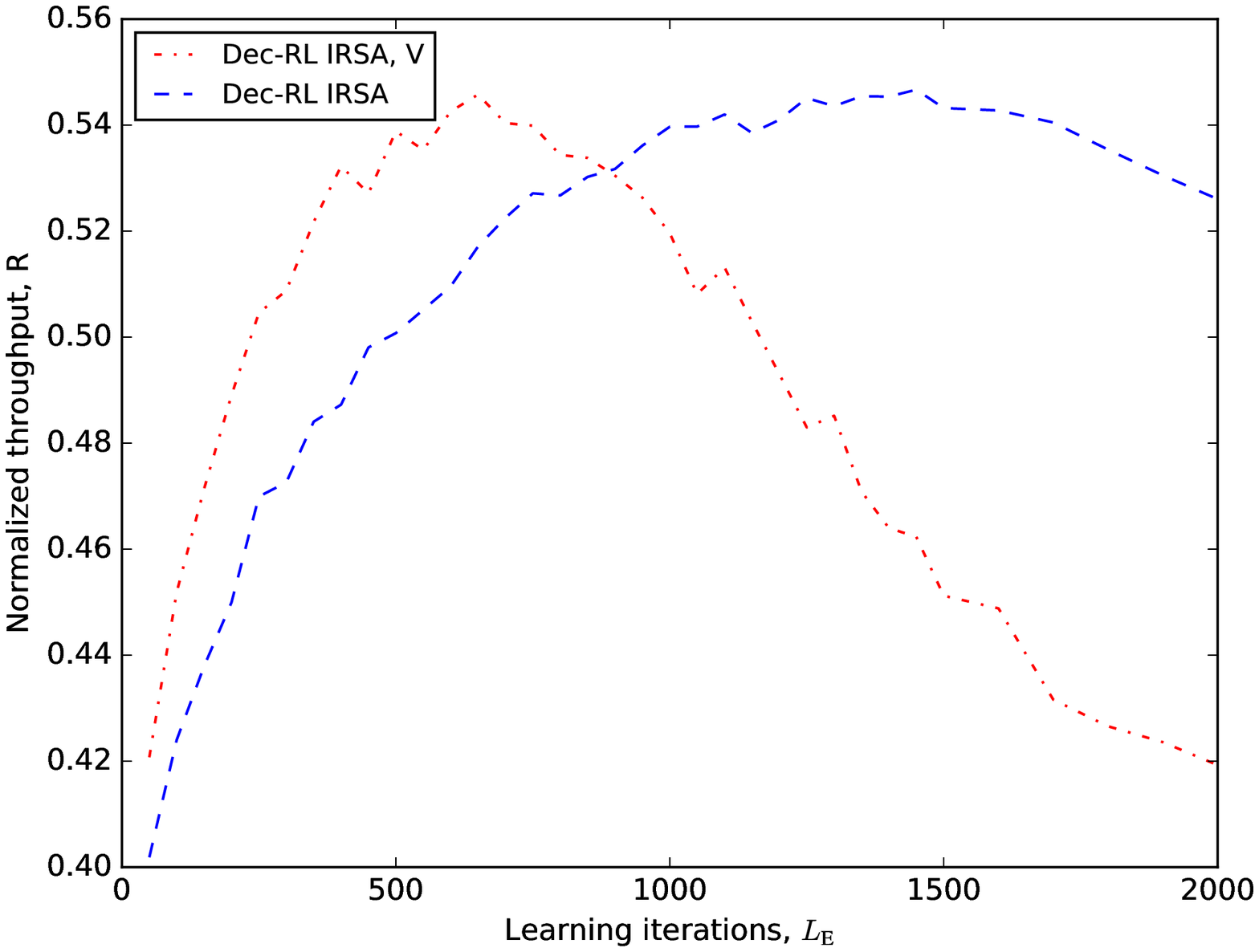}
\caption{Comparison of throughput of virtual and vanilla Dec-RL IRSA for different number of learning iterations and $G=0.7$.}
\label{fig:virtual}
\end{minipage}
\end{figure}
\begin{figure}[t]
\centering
\begin{minipage}[t]{.45\textwidth}
\includegraphics[scale=0.4]{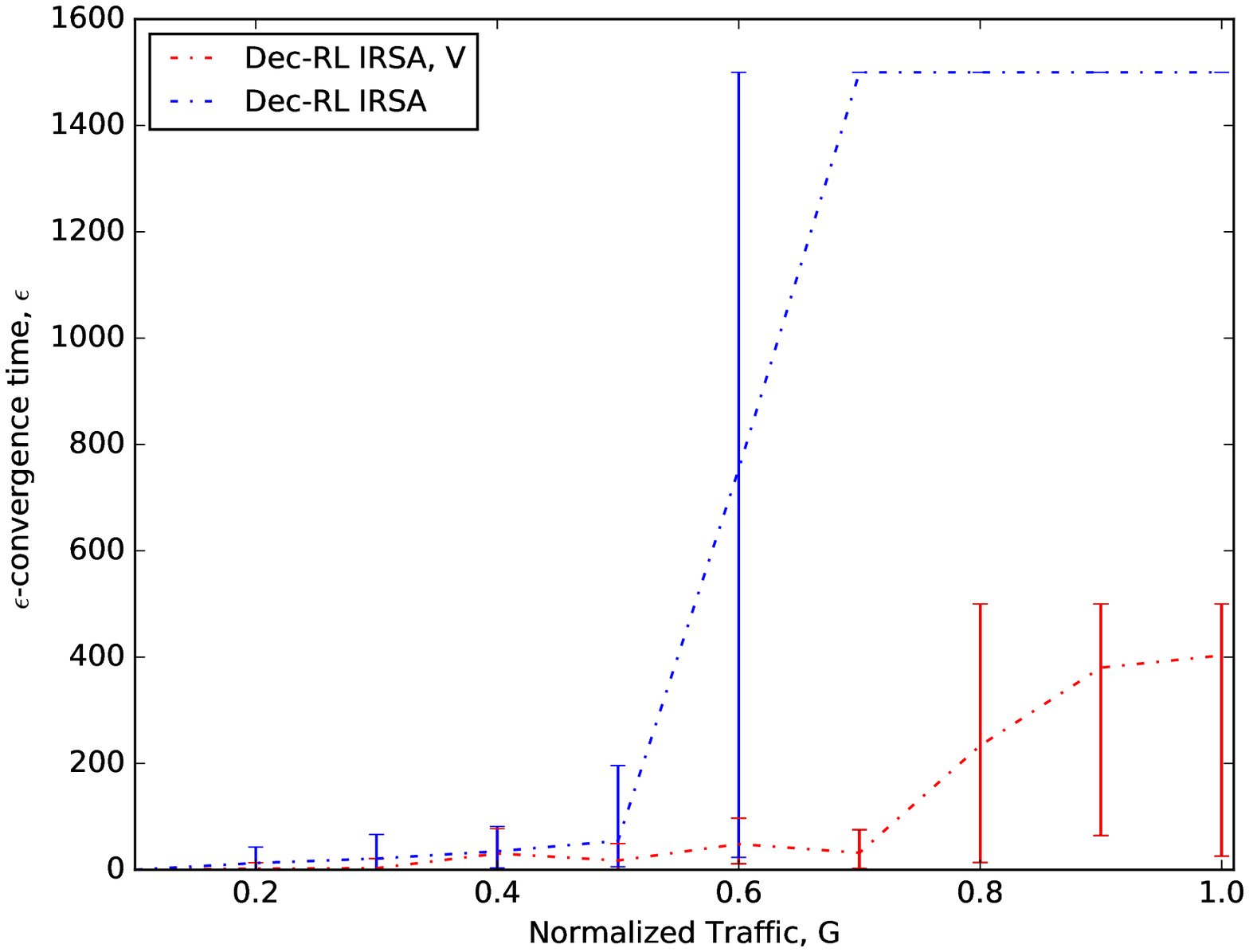}
\caption{Statistical comparison of $\epsilon$-convergence times for vanilla Dec-RL IRSA and Dec-RL IRSA using virtual experience, with $\epsilon =0.5$.}
\vspace{-0.8cm}
\label{fig:statistical_virtual}
\end{minipage} \hfill
\begin{minipage}[t]{.45\textwidth}
\includegraphics[scale = 0.4]{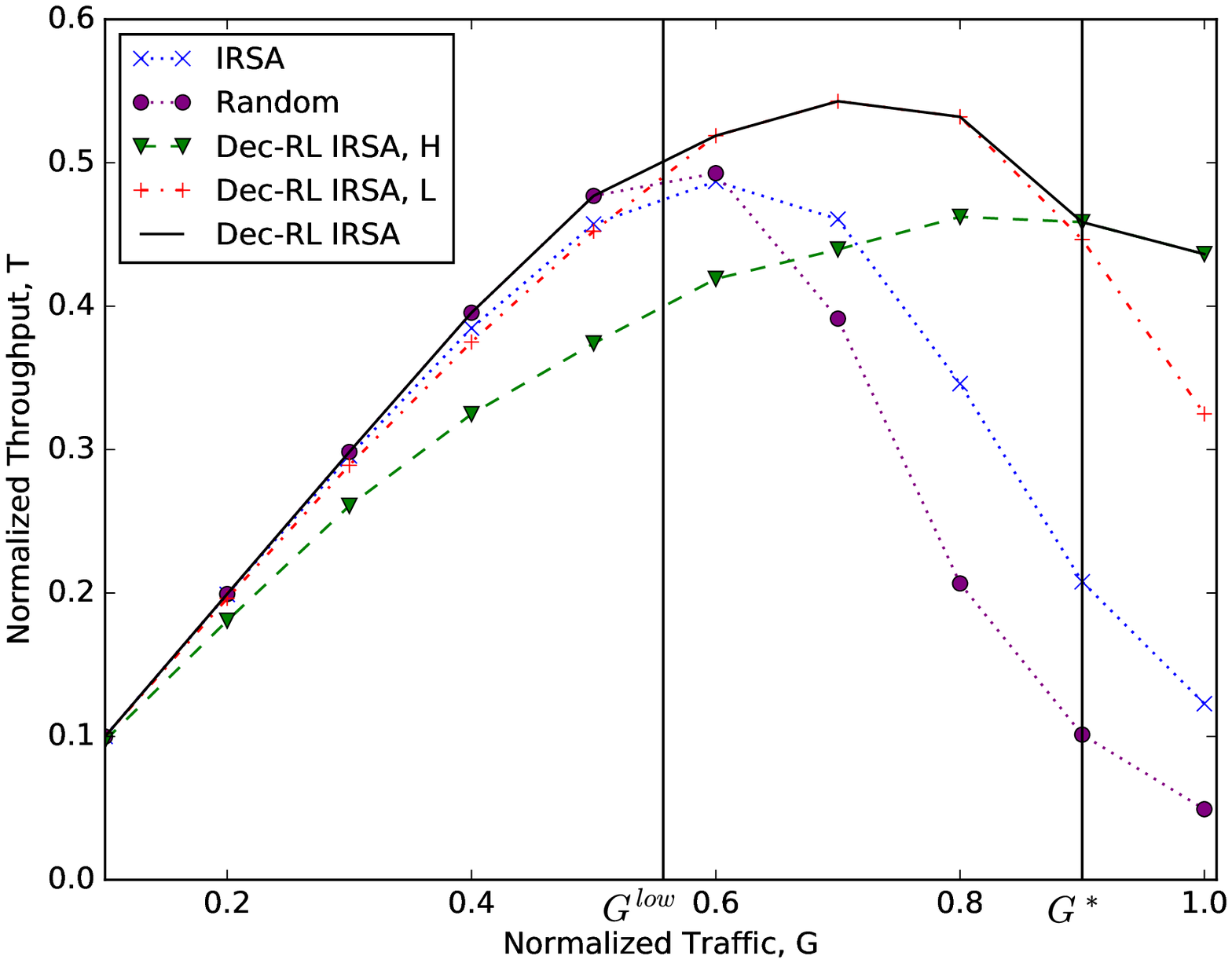}
\caption{Performance comparison Dec-RL IRSA, a random strategy, Dec-RL IRSA, L optimized for low $G$ and Dec-RL IRSA, H optimized for high values of $G$, and the vanilla IRSA scheme. The solid line indicates the throughput when considering the optimal scheme for each value of $G$.}
\label{fig:waterfall}
\end{minipage}
\end{figure}

\section{Conclusion}\label{sec:discussion}

We have examined the problem of decentralized MAC design through a reinforcement learning perspective and proved that learning transmission strategies can be beneficial even under the assumption of sensor nodes' independent learning. Our experiments suggest that the ``waterfall effect" of the problem, common in social games where agents compete for common resources, leads to different learning dynamics that demand adaptive solutions. Our method's superiority is manifested especially in high channel loads, where the need for adaptivity is more eminent and agents benefit from short-sightedness and increased exploration, which implicitly ensures better coordination. From the results we can conclude that in  order to make learning tractable for online application scenarios, it is essential to achieve fast convergence. We observed that even maintaining a small observation space, by restricting the history window to 2, the performance remains satisfactory. Finally, the results show that we significantly reduced convergence time by introducing virtual experience into learning. 


\bibliographystyle{IEEEtran}

\begin{thebibliography}{10}
\providecommand{\url}[1]{#1}
\csname url@samestyle\endcsname
\providecommand{\newblock}{\relax}
\providecommand{\bibinfo}[2]{#2}
\providecommand{\BIBentrySTDinterwordspacing}{\spaceskip=0pt\relax}
\providecommand{\BIBentryALTinterwordstretchfactor}{4}
\providecommand{\BIBentryALTinterwordspacing}{\spaceskip=\fontdimen2\font plus
\BIBentryALTinterwordstretchfactor\fontdimen3\font minus
  \fontdimen4\font\relax}
\providecommand{\BIBforeignlanguage}[2]{{%
\expandafter\ifx\csname l@#1\endcsname\relax
\typeout{** WARNING: IEEEtran.bst: No hyphenation pattern has been}%
\typeout{** loaded for the language `#1'. Using the pattern for}%
\typeout{** the default language instead.}%
\else
\language=\csname l@#1\endcsname
\fi
#2}}
\providecommand{\BIBdecl}{\relax}
\BIBdecl

\bibitem{Abramson1970THEAS}
N.~M. Abramson, ``{THE ALOHA SYSTEM: Another Alternative for Computer
  Communications},'' in \emph{Proc. of joint Computing Conf. AFIPS'70},
  Honolulu, HI, USA, Nov. 1970.

\bibitem{Hadded2015}
M.~Hadded, P.~Muhlethaler, A.~Laouiti, R.~Zagrouba, and L.~A. Saidane,
  ``Tdma-based mac protocols for vehicular ad hoc networks: A survey,
  qualitative analysis, and open research issues,'' \emph{IEEE Communications
  Surveys {\&} Tutorials}, vol.~17, no.~4, pp. 2461--2492, Jun. 2015.

\bibitem{1673243}
Z.~Liu and I.~Elhanany, ``{RL-MAC: A QoS-Aware Reinforcement Learning based MAC
  Protocol for Wireless Sensor Networks},'' in \emph{Proc. IEEE Int. Conf. on
  Networking, Sensing and Control, ICNSC '06}, Ft. Lauderdale, FL, USA, Aug.
  2006.

\bibitem{Choudhury1983}
G.~L. Choudhury and S.~S. Rappaport, ``{Diversity ALOHA \- A Random Access
  Scheme for Satellite Communications},'' \emph{IEEE Trans. on Communications},
  vol.~31, no.~3, pp. 450--457, Mar. 1983.

\bibitem{4155680}
E.~Casini, R.~D. Gaudenzi, and O.~D.~R. Herrero, ``Contention resolution
  diversity slotted aloha (crdsa): An enhanced random access schemefor
  satellite access packet networks,'' \emph{IEEE Trans. on Wireless
  Communications}, vol.~6, no.~4, pp. 1408--1419, Apr. 2007.

\bibitem{5668922}
G.~Liva, ``Graph-based analysis and optimization of contention resolution
  diversity slotted aloha,'' \emph{IEEE Trans. on Communications}, vol.~59,
  no.~2, pp. 477--487, Feb. 2011.

\bibitem{7302046}
E.~Paolini, G.~Liva, and M.~Chiani, ``Coded slotted aloha: A graph-based method
  for uncoordinated multiple access,'' \emph{IEEE Trans. on Information
  Theory}, vol.~61, no.~12, pp. 6815--6832, Dec. 2015.

\bibitem{6503624}
A.~Meloni, M.~Murroni, C.~Kissling, and M.~Berioli, ``Sliding window-based
  contention resolution diversity slotted aloha,'' in \emph{Proc. of IEEE
  Global Communications Conf., GLOBECOM'12}, Anaheim, CA, USA, Dec. 2012.

\bibitem{7762138}
E.~Sandgren, A.~G. i~Amat, and F.~Brannstrom, ``On frame asynchronous coded
  slotted aloha: Asymptotic, finite length, and delay analysis,'' \emph{IEEE
  Trans. on Communications}, vol.~65, no.~2, pp. 691--704, Feb 2017.

\bibitem{Toni2018}
\BIBentryALTinterwordspacing
L.~Toni and P.~Frossard, ``{IRSA Transmission Optimization via Online
  Learning},'' 2018. [Online]. Available: \url{http://arxiv.org/abs/1801.09060}
\BIBentrySTDinterwordspacing

\bibitem{7305777}
------, ``Prioritized random mac optimization via graph-based analysis,''
  \emph{IEEE Trans. on Communications}, vol.~63, no.~12, pp. 5002--5013, Dec.
  2015.

\bibitem{Waugh:2009:APE:1558109.1558119}
K.~Waugh, D.~Schnizlein, M.~Bowling, and D.~Szafron, ``{Abstraction Pathologies
  in Extensive Games},'' in \emph{Proc. of AAMAS '09}, Budapest, Hungary, May
  2009.

\bibitem{Bellman:2003:DP:862270}
R.~E. Bellman, \emph{Dynamic Programming}.\hskip 1em plus 0.5em minus
  0.4em\relax Dover Publications, Incorporated, 2003.

\bibitem{AAAI113750}
C.~Zhang and V.~Lesser, ``Coordinated multi-agent reinforcement learning in
  networked distributed pomdps,'' San Francisco, CA, USA, Aug. 2011.

\bibitem{DBLP:journals/corr/abs-1301-3836}
D.~S. Bernstein, S.~Zilberstein, and N.~Immerman, ``{The Complexity of
  Decentralized Control of Markov Decision Processes},'' \emph{Mathematics of
  Operations Research}, vol.~27, no.~4, pp. 819--840, Nov. 2002.

\bibitem{Lesser:2003:DSN:940763}
V.~Lesser, M.~Tambe, and C.~L. Ortiz, Eds., \emph{Distributed Sensor Networks:
  A Multiagent Perspective}.\hskip 1em plus 0.5em minus 0.4em\relax Norwell,
  MA, USA: Kluwer Academic Publishers, 2003.

\bibitem{1420665}
J.~Dowling, E.~Curran, R.~Cunningham, and V.~Cahill, ``Using feedback in
  collaborative reinforcement learning to adaptively optimize manet routing,''
  \emph{IEEE Trans. on Systems, Man, and Cybernetics - Part A: Systems and
  Humans}, vol.~35, no.~3, pp. 360--372, May 2005.

\bibitem{Nair:2005:NDP:1619332.1619356}
R.~Nair, P.~Varakantham, M.~Tambe, and M.~Yokoo, ``Networked distributed
  pomdps: A synthesis of distributed constraint optimization and pomdps,'' in
  \emph{Proc. of the 20th National Conference on Artificial Intelligence,
  AAAI'05}, Jul. 2005.

\bibitem{DBLP:journals/corr/abs-1710-08803}
T.~Park and W.~Saad, ``Distributed learning for low latency machine type
  communication in a massive internet of things,'' \emph{CoRR}, vol.
  abs/1710.08803, 2017.

\bibitem{MnihKSGAWR13}
V.~Mnih, K.~Kavukcuoglu, D.~Silver, A.~Graves, I.~Antonoglou, D.~Wierstra, and
  M.~A. Riedmiller, ``Playing atari with deep reinforcement learning,'' in
  \emph{Proc. of NIPS Deep Learning Workshop, NIPS'13}, Lake Tahoe, CA, USA,
  Dec. 2013.

\bibitem{4066245}
M.~Minsky, ``Steps toward artificial intelligence,'' \emph{Proceedings of the
  IRE}, vol.~49, no.~1, pp. 8--30, Jan. 1961.

\bibitem{4221491}
T.~Chen, H.~Zhang, G.~M. Maggio, and I.~Chlamtac, ``Cogmesh: A cluster-based
  cognitive radio network,'' in \emph{Proc. of IEEE Int. Symp. on New Frontiers
  in Dynamic Spectrum Access Networks, DySPAN'07}, Apr. 2007.

\bibitem{7247494}
E.~Paolini, ``Finite length analysis of irregular repetition slotted aloha
  (irsa) access protocols,'' in \emph{Proc. of IEEE Int. Conf. on Communication
  Workshop, ICCW'15}, London, UK, Jun. 2015.

\bibitem{5986747}
N.~Mastronarde and M.~van~der Schaar, ``{Fast Reinforcement Learning for
  Energy-Efficient Wireless Communication},'' \emph{IEEE Trans. on Signal
  Processing}, vol.~59, no.~12, pp. 6262--6266, Dec. 2011.

\bibitem{Littman:1994:MGF:3091574.3091594}
M.~L. Littman, ``Markov games as a framework for multi-agent reinforcement
  learning,'' in \emph{Proc. of the 11th Int. Conf. on Machine Learning,
  ICML'94}, New Brunswick, NJ, USA, Jul. 1994.

\bibitem{Watkins1992}
C.~J.~C.~H. Watkins and P.~Dayan, ``Q-learning,'' \emph{Machine Learning},
  vol.~8, no.~3, pp. 279--292, May 1992.

\bibitem{4480049}
P.~Nurmi, ``Reinforcement learning for routing in ad hoc networks,'' in
  \emph{Proc. of IEEE Int. Symp. on Modeling and Optimization in Mobile, Ad Hoc
  and Wireless Networks and Workshops, WiOpt'07}, Limassol, Cyprus, Apr. 2007.

\bibitem{Kaelbling:1998:PAP:1643275.1643301}
L.~P. Kaelbling, M.~L. Littman, and A.~R. Cassandra, ``Planning and acting in
  partially observable stochastic domains,'' \emph{Artif. Intell.}, vol. 101,
  no. 1-2, pp. 99--134, May 1998.

\bibitem{Claus:1998:DRL:295240.295800}
C.~Claus and G.~Boutilier, ``{The Dynamics of Reinforcement Learning in
  Cooperative Multiagent Systems},'' in \emph{Proc. of the 15th National/10th
  Int. Conf. on Artificial Intelligence/Innovative Applications of Artificial
  Intelligence, AAAI '98/IAAI '98}, Madison, WI, USA, Jul. 1998.

\bibitem{DBLP:journals/corr/ThomosKFS14}
N.~Thomos, E.~Kurdoglu, P.~Frossard, and M.~van~der Schaar, ``Adaptive
  prioritized random linear coding and scheduling for layered data delivery
  from multiple servers,'' \emph{IEEE Transactions on Multimedia}, vol.~17,
  no.~6, pp. 893--906, June 2015.

\bibitem{MastroThesis}
N.~H. Mastronarde, ``Online learning for energy-efficient multimedia systems,''
  Ph.D. dissertation, University of California, 2011.

\bibitem{Even-Dar:2004:LRQ:1005332.1005333}
E.~Even-Dar and Y.~Mansour, ``Learning rates for q-learning,'' \emph{J. Mach.
  Learn. Res.}, vol.~5, pp. 1--25, Dec. 2004.

\bibitem{Storn1997}
R.~Storn and K.~Price, ``{Differential Evolution -- A Simple and Efficient
  Heuristic for global Optimization over Continuous Spaces},'' \emph{Journal of
  Global Optimization}, vol.~11, no.~4, pp. 341--359, Dec. 1997.

\bibitem{910578}
T.~J. Richardson, M.~A. Shokrollahi, and R.~L. Urbanke, ``Design of
  capacity-approaching irregular low-density parity-check codes,'' \emph{IEEE
  Trans. on Information Theory}, vol.~47, no.~2, pp. 619--637, Feb 2001.

\bibitem{GrenagerAAAI02}
T.~Grenager, R.~Powers, and Y.~Shoham, ``{Dispersion Games: General Definitions
  and Some Specific Learning Results},'' in \emph{Proc. of 18th National/14th
  Conf. on Artificial Intelligence/Innovative Applications of Artificial
  Intelligence, AAAI '02}, Edmonton, AL, Canada, Jul. 2002.

\end{thebibliography}

\end{document}